\documentclass[]{interact}

\usepackage{amsmath}
\usepackage{amssymb}
\usepackage{booktabs}
\usepackage{tabularx}
\usepackage{array}
\usepackage{placeins}
\usepackage{url}
\usepackage[table]{xcolor}
\usepackage{tikz}
\usepackage{pgfplots}
\usepackage{graphicx}
\usepackage{algorithm}
\usepackage{algpseudocode}
\pgfplotsset{compat=1.18}

\definecolor{PlotBlue}{RGB}{0,114,178}
\definecolor{PlotOrange}{RGB}{230,159,0}
\definecolor{PlotGreen}{RGB}{0,158,115}
\definecolor{PlotPurple}{RGB}{204,121,167}

\newcolumntype{Y}{>{\raggedright\arraybackslash}X}

\usepackage{natbib}
\bibpunct[, ]{(}{)}{;}{a}{}{,}
\renewcommand\bibfont{\fontsize{10}{12}\selectfont}
\usepackage[hidelinks,hypertexnames=false]{hyperref}
\hypersetup{
  pdftitle={Cycle time minimization for the simple assembly line balancing problem under peak power constraints},
  pdfauthor={Bao Gia Hoang, Tuyen Van Kieu, and Khanh Van To},
  pdfsubject={Research Article submitted to Engineering Optimization},
  pdfkeywords={assembly line balancing; cycle time minimization; peak power constraint; satisfiability; exact optimization}
}

\theoremstyle{plain}
\newtheorem{lemma}{Lemma}[section]
\newtheorem{proposition}[lemma]{Proposition}

\newcommand{\Problem}{SALBP-PPC}
\newcommand{\Qmax}{Q_{\max}}
\newcommand{\Cinit}{C^{\mathrm{init}}}
\newcommand{\LBzero}{\ensuremath{LB_0}}
\newcommand{\Status}[1]{\ensuremath{\mathsf{#1}}}

\begin{document}

\articletype{ARTICLE}

\title{Cycle time minimization for the simple assembly line balancing problem under peak power constraints}

\author{
\name{Bao Gia Hoang\textsuperscript{a}, Tuyen Van Kieu\textsuperscript{a}, and Khanh Van To\textsuperscript{a}\thanks{CONTACT Khanh Van To. Email: khanhtv@vnu.edu.vn}}
\affil{\textsuperscript{a}VNU University of Engineering and Technology, Hanoi, Vietnam}
\affil{Institutional emails: Bao Gia Hoang, \href{mailto:23020653@vnu.edu.vn}{23020653@vnu.edu.vn}; Tuyen Van Kieu, \href{mailto:tuyenkv@vnu.edu.vn}{tuyenkv@vnu.edu.vn}; Khanh Van To, \href{mailto:khanhtv@vnu.edu.vn}{khanhtv@vnu.edu.vn}}
\affil{ORCID iDs: Bao Gia Hoang, \href{https://orcid.org/0009-0004-1400-7358}{0009-0004-1400-7358}; Tuyen Van Kieu, \href{https://orcid.org/0009-0007-7800-7172}{0009-0007-7800-7172}; Khanh Van To, \href{https://orcid.org/0009-0008-1907-7848}{0009-0008-1907-7848}}
}

\maketitle

\noindent{\footnotesize\textbf{Word count:} 6,823 words, including tables, references, and figure captions.}\par\medskip

\begin{abstract}
Peak power limits restrict concurrent tasks and may increase assembly-line cycle time. To the best of our knowledge, this study is the first to minimize cycle time for the simple assembly line balancing problem type 2 (SALBP-2) with a fixed number of workstations and a fixed limit on total instantaneous power. An exact satisfiability (SAT) method finds feasible schedules, searches systematically for shorter cycles, and proves optimality when possible. Two reproducible formulas define looser and tighter power limits for 72 cases based on the standard SALBP library, with Gurobi and CPLEX providing commercial MIP and CP comparisons. Relative to standard SALBP-2 optima, the two limits increase the best-known cycle time by 16.49\% and 61.02\% on average. The best-performing SAT configurations find a feasible solution for every case, solve more cases to optimality than each commercial solver, and when both prove optimality, are almost always faster under the reported settings.
\end{abstract}

\begin{keywords}
assembly line balancing; cycle time minimization; peak power constraint; satisfiability; exact optimization
\end{keywords}

\section{Introduction}

The simple assembly line balancing problem (SALBP) assigns precedence-constrained tasks to ordered workstations. In SALBP-2, the number of workstations is fixed and the cycle time is minimized, thereby maximizing the production rate. Exact and heuristic procedures for SALBP-2 continue to be developed \citep{Scholl2006,Li2021Enhanced,AlvarezMiranda2024BBR}.

Energy-aware line balancing includes power or energy data for each task. The simple assembly line balancing problem with power-peak minimization (SALB3PM) fixes the number of workstations and the cycle time, then minimizes the highest total power drawn by the line at any time \citep{Gianessi2019Simple}. Later mathematical-programming, heuristic, satisfiability (SAT), and maximum-satisfiability methods reduce this peak without changing the cycle time \citep{Delorme2023A,Py2024POS,Zheng2024Optimizing,VanKieu2025Compact}.

A fixed peak power limit creates a different problem: given the number of workstations and the limit, the line designer seeks the shortest feasible cycle. Machine-scheduling studies also limit total power at each time while minimizing the overall completion time, or makespan \citep{Fang2013Flow,Wang2019Decoding,Lv2024Considering,Yuan2026Flexible}. However, they use job routes or machine assignments instead of assigning precedence-constrained tasks to ordered SALBP workstations.

This leaves an open problem in assembly-line balancing. To the best of our knowledge, no previous SALBP study has minimized cycle time when both the number of workstations and the limit on total instantaneous power are fixed. We call this problem the \emph{simple assembly line balancing problem with peak power constraints} (\Problem).

Building on earlier SAT models \citep{Py2024POS,VanKieu2025Compact}, this study makes three contributions. First, it defines cycle time minimization under a fixed peak power limit and gives two reproducible formulas for setting the test limits. Second, it develops an exact SAT method that finds a feasible cycle, repeatedly searches for a shorter one, and proves optimality when possible. Third, it presents three ways to prevent tasks at the same workstation from overlapping and compares the resulting SAT methods with commercial mixed-integer programming (MIP) and constraint programming (CP) solvers on 72 benchmark cases.

\section{Related work}

Classical SALBP-2 fixes the number of workstations and minimizes cycle time. \citet{Scholl2006} review its exact and heuristic methods. Among recent exact methods, \citet{Li2021Enhanced} repeatedly check related SALBP-1 problems, whereas \citet{AlvarezMiranda2024BBR} develop a direct branch, bound and remember algorithm for SALBP-2. Neither method models task power or the power drawn by tasks running at the same time.

The type-II resource-constrained assembly line balancing problem has the closest objective. It also fixes the number of workstations and minimizes cycle time, but it assigns limited equipment or other task-specific resources to workstations \citep{Michels2022Resource}. Its resource limits depend on workstation assignment. In contrast, a peak power limit depends on which tasks run at the same time across all workstations.

Energy-aware assembly-line studies optimize energy alone or together with production measures in robotic, mixed-model, parallel, and controlled lines \citep{Ramli2022Review}. Their objectives include cycle time and total energy \citep{Nilakantan2015Energy,Kubilay2026Constraint}, energy use and balance rate \citep{Zhang2020A}, or total cycle time, energy cost, and grid peaks \citep{Liberati2022Energy}. None minimizes cycle time alone under a fixed limit on total line power. Table~\ref{tab:related-work} summarizes the differences.

\begin{table}[t!]
\centering
\caption{Comparison of related problem classes by objective, peak-power treatment, and scheduling structure.}
\label{tab:related-work}
\footnotesize
\renewcommand{\arraystretch}{1.18}
\setlength{\tabcolsep}{3pt}
\begin{tabularx}{\textwidth}{@{}
>{\raggedright\arraybackslash\hsize=0.88\hsize}X
>{\raggedright\arraybackslash\hsize=0.72\hsize}X
>{\raggedright\arraybackslash\hsize=0.95\hsize}X
>{\raggedright\arraybackslash\hsize=1.10\hsize}X
>{\raggedright\arraybackslash\hsize=1.35\hsize}X
@{}}
\toprule
\shortstack[l]{Production\\setting} &
Objective(s) &
\shortstack[l]{Peak-power\\treatment} &
\shortstack[l]{Assignment or\\routing structure} &
\shortstack[l]{Representative\\studies} \\
\midrule
\multicolumn{5}{@{}p{\textwidth}@{}}{\textit{Assembly-line balancing literature}} \\
\addlinespace[2pt]
Simple assembly line & Cycle time & Not modeled & Ordered workstation assignment
& \citep{Scholl2006,Li2021Enhanced,AlvarezMiranda2024BBR} \\
\addlinespace[1.5pt]
Resource-constrained assembly line & Cycle time & No limit on total instantaneous power & Ordered workstations with assignable resources
& \citep{Michels2022Resource} \\
\addlinespace[1.5pt]
Energy-aware assembly lines & Cycle time and/or energy measures & Energy use or grid peaks enter the objective; no fixed limit on total line power & Robotic, mixed-model, parallel, or controlled lines
& \citep{Nilakantan2015Energy,Zhang2020A,Kubilay2026Constraint,Liberati2022Energy} \\
\addlinespace[1.5pt]
Simple assembly line with task timing & Peak power (fixed $C$) & Peak power is minimized & Ordered workstation assignment and task start times
& \citep{Gianessi2019Simple,Lamy2020IFAC,Py2024POS,Zheng2024Optimizing,VanKieu2025Compact} \\
\addlinespace[4pt]
\multicolumn{5}{@{}p{\textwidth}@{}}{\textit{Adjacent production-line scheduling literature}} \\
\addlinespace[2pt]
Asynchronous production line & Production during demand response & Limit during demand-response intervals & Changes between production states
& \citep{Desta2018Demand} \\
\addlinespace[4pt]
\multicolumn{5}{@{}p{\textwidth}@{}}{\textit{Peak-constrained machine scheduling literature}} \\
\addlinespace[2pt]
Flow shop & Makespan & Fixed peak power limit & Fixed machine routes
& \citep{Fang2013Flow,Wang2019Decoding,Lv2024Considering} \\
\addlinespace[1.5pt]
Job-shop-like system & Total completion time & Power limit changes over time & Job-specific machine routes
& \citep{Kemmoe2017Jobshop} \\
\addlinespace[1.5pt]
Job shop & Electricity cost & Peak power limit & Job-specific machine routes
& \citep{Masmoudi2019Jobshop} \\
\addlinespace[1.5pt]
Job shop & Makespan and/or energy & Peak power limit & Job-specific machine routes
& \citep{Carlucci2023JobShop,Homayouni2025Optimizing} \\
\addlinespace[1.5pt]
Flexible job shop & Makespan; or makespan, energy, and cost & Fixed peak power limit & Alternative machine assignments and job routes
& \citep{Wang2024Multi-objective,Yuan2026Flexible} \\
\addlinespace[4pt]
\multicolumn{5}{@{}p{\textwidth}@{}}{\textit{Combination considered in this study}} \\
\addlinespace[2pt]
\rowcolor{gray!10}Simple assembly line with task timing & \textbf{Cycle time} & \textbf{Fixed limit on total line power} & Ordered workstation assignment and task start times & \textbf{This study} \\
\bottomrule
\end{tabularx}
\end{table}

\FloatBarrier

SALB3PM links workstation assignment, task start times, and peak power. Later studies consider no-idle task sequences \citep{Lamy2020IFAC}, reconfigurable systems and solutions represented by task permutations \citep{Delorme2024Line,Delorme2023A}, SAT and MaxSAT methods \citep{Py2024POS,Zheng2024Optimizing}, a compact SAT formulation \citep{VanKieu2025Compact}, and search from several starting solutions \citep{Araujo2025SALB3PM}. Unlike \Problem{}, these models fix cycle time and minimize peak power \citep{Gianessi2019Simple}.

Power-constrained scheduling provides a second comparison. Asynchronous-line models schedule changes in machine operating states under demand-response limits \citep{Desta2018Demand}. Flow-shop studies minimize makespan under limits on total power at each time \citep{Fang2013Flow,Wang2019Decoding,Lv2024Considering}. Job-shop models combine fixed or time-varying power limits with completion-time, electricity-cost, makespan, or energy objectives \citep{Kemmoe2017Jobshop,Masmoudi2019Jobshop,Carlucci2023JobShop,Homayouni2025Optimizing}. Flexible job-shop models minimize makespan alone or together with energy and cost \citep{Yuan2026Flexible,Wang2024Multi-objective}. These models use job-specific routes and rules about which machines can perform each operation; they do not assign tasks to ordered SALBP workstations.

No existing model combines all of these features. Classical and resource-constrained type-II assembly-line models fix the number of workstations and minimize cycle time, but do not limit the total power of tasks running at the same time. Energy-aware assembly-line and SALB3PM models instead optimize energy or peak power, or keep cycle time fixed. Machine-scheduling models can limit instantaneous power, but do not use ordered SALBP workstation assignment. \Problem{} combines that assignment rule with a fixed number of workstations and a fixed power limit while minimizing cycle time.

\FloatBarrier

\section{Problem definition}

Let $N=\{1,\ldots,n\}$ be the task set and $K=\{1,\ldots,m\}$ the ordered workstation set, where $n,m\in\mathbb{Z}_{>0}$. Below, \emph{station} is used as a shorter synonym for workstation. Task $i\in N$ has integer duration $t_i\in\mathbb{Z}_{>0}$ and constant integer power demand $w_i\in\mathbb{Z}_{>0}$. The precedence graph is a directed acyclic graph with arc set $E\subseteq N\times N$: $(i,j)\in E$ means that task $i$ must come before task $j$. The number of workstations $m$ and the total-power limit $\Qmax\in\mathbb{Z}_{>0}$ are fixed, whereas the cycle time $C\in\mathbb{Z}_{>0}$ is minimized.

\begin{table}[H]
\centering
\caption{Problem and SAT-model notation.}
\label{tab:notation}
\footnotesize
\renewcommand{\arraystretch}{1.15}
\begin{tabularx}{\textwidth}{@{} l Y @{}}
\toprule
Symbol & Meaning \\
\midrule
\multicolumn{2}{@{}l}{\textit{Schedule and peak-power limit}} \\
\addlinespace[2pt]
$a(i)$ & Workstation assigned to task $i$. \\
$\sigma_i$ & Integer start time of task $i$ within the cycle. \\
$C$ & Positive-integer cycle time; the objective is to minimize it. \\
$H$; $H_i$ & Time slots $\{0,\ldots,C-1\}$; feasible starts $\{0,\ldots,C-t_i\}$. \\
$LB_Q$; $UB_Q$ & Lower and upper bounds used to construct the peak power limit. \\
\addlinespace[4pt]
\multicolumn{2}{@{}l}{\textit{SAT variables}} \\
\addlinespace[2pt]
$X_{i,k}$ & Task $i$ is assigned to workstation $k$. \\
$S_{i,s}$ & Task $i$ starts at time $s$. \\
$A_{i,\tau}$ & Task $i$ is active in slot $\tau$. \\
$R_{i,k}$ & Task $i$ is assigned to station $k$ or earlier. \\
$T_{i,u}$ & Task $i$ starts no later than time $u$. \\
$D_{i,j}$ & Tasks $i$ and $j$ share a workstation (SS and SS-D only). \\
\bottomrule
\end{tabularx}
\end{table}

A schedule consists of a workstation assignment $a:N\rightarrow K$ and start times $\sigma_i\in\mathbb{Z}_{\geq 0}$, $i\in N$. Task $i$ runs from time $\sigma_i$ up to, but not including, time $\sigma_i+t_i$; this interval is written $[\sigma_i,\sigma_i+t_i)$. The schedule is feasible for cycle time $C$ if
\begin{align}
& a(i)\in K,\quad 0\leq\sigma_i\leq C-t_i && i\in N, \label{eq:domain}\\
& a(i)\leq a(j) && (i,j)\in E, \label{eq:station-prec}\\
& a(i)=a(j)\Rightarrow \sigma_i+t_i\leq\sigma_j && (i,j)\in E, \label{eq:time-prec}\\
& a(i)=a(j)\Rightarrow\sigma_j\geq\sigma_i+t_i\text{ or }\sigma_i\geq\sigma_j+t_j && i<j,\,i,j\in N. \label{eq:nonoverlap}
\end{align}
Condition \eqref{eq:domain} places every task at one workstation and within $[0,C)$. Conditions \eqref{eq:station-prec} and \eqref{eq:time-prec} enforce the required station order and, for related tasks at the same station, their time order. Condition \eqref{eq:nonoverlap} prevents tasks at the same station from overlapping. Define the activity indicator
\[
I_{i,\tau}=\begin{cases}1,&\sigma_i\leq\tau<\sigma_i+t_i,\\0,&\text{otherwise},\end{cases}
\qquad i\in N,\ \tau\in H.
\]
The peak power constraint (PPC) limits the total power of all active tasks in every slot:
\begin{equation}
\sum_{i\in N}w_i I_{i,\tau}\leq\Qmax
\qquad \tau\in H.
\label{eq:power-cap}
\end{equation}
Together, \eqref{eq:domain}--\eqref{eq:power-cap} define a feasible \Problem{} schedule.
The objective is to minimize $C$ subject to these constraints. By \eqref{eq:domain}, every feasible solution satisfies $C\geq\max_{i\in N}(\sigma_i+t_i)$. If $C$ is larger than this value, it can be reduced to the time when the last task finishes without changing any assignment, precedence, non-overlap, or power constraint. Hence an optimal solution satisfies
\[
C=\max_{i\in N}(\sigma_i+t_i),
\]
so minimizing cycle time is the same as minimizing the time when the last task finishes.

\subsection{Illustrative schedule}

Consider five tasks with $(t_i,w_i)=(2,5),(4,3),(3,6),(3,4),(2,5)$, listed by task index, and precedence chains $1\prec3\prec5$ and $2\prec4$. Figure~\ref{fig:example} shows the assignment $a=(1,2,1,2,1)$ and start-time vector $\sigma=(0,0,2,4,5)$. Tasks at workstations 1 and 2 therefore occupy $[0,2),[2,5),[5,7)$ and $[0,4),[4,7)$, respectively. Both precedence chains are respected, and tasks do not overlap within either workstation. The resulting line power is 8 on $[0,2)$, 9 on $[2,4)$, 10 on $[4,5)$, and 9 on $[5,7)$. Thus its peak is $10<\Qmax=11$ and its cycle time is $C=7$.

\begin{figure}[H]
\centering
\includegraphics[width=0.85\textwidth]{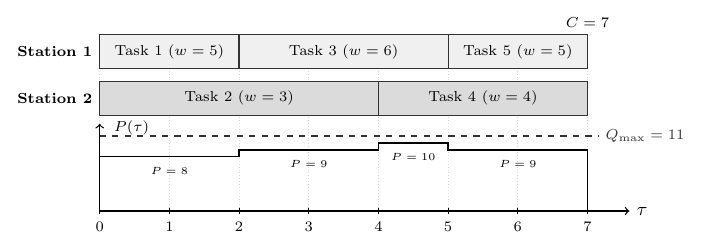}
\caption{A feasible two-station schedule with cycle time $C=7$. Its total power never exceeds $10$, below $\Qmax=11$.}
\label{fig:example}
\end{figure}

\section{Optimization models and solution methods}

For fixed $C$, \Problem{} asks whether a feasible schedule exists. This section presents the SAT model, three non-overlap formulations, the search for the smallest $C$, and the MIP and CP comparison models.

\subsection{SAT feasibility model for a fixed cycle time}

Any $C<\max_i t_i$ is infeasible, so assume $C\geq\max_i t_i$. Let $H=\{0,\ldots,C-1\}$ and $H_i=\{0,\ldots,h_i\}$, where $h_i=C-t_i$. Boolean variables $X_{i,k}$, $S_{i,s}$, and $A_{i,\tau}$ state whether task $i$ uses station $k$, starts at $s$, and is active in slot $\tau$. Cumulative variables record choices at or before a given position:
\begin{align*}
R_{i,k}&\equiv\bigvee_{\ell=1}^{k}X_{i,\ell}
&& i\in N,\ k=1,\ldots,m-1,\\
T_{i,u}&\equiv\bigvee_{s=0}^{u}S_{i,s}
&& i\in N,\ u=0,\ldots,h_i-1.
\end{align*}
Thus, $R_{i,k}$ means station $k$ or earlier, and $T_{i,u}$ means start time $u$ or earlier. They adapt the compact SAT encoding (CSE) for SALB3PM \citep{VanKieu2025Compact}, using a linear-size sequential encoding for exactly-one choices \citep{sinz_sequential}. Let $\Phi(C,\Qmax)$ denote the model. Its pseudo-Boolean (PB) power constraints are weighted sums of Boolean variables.

For $m\geq2$, exactly one station is selected by
\begin{alignat}{2}
& R_{i,1}\Leftrightarrow X_{i,1}, && i\in N, \tag{CSE-1}\\
& R_{i,k-1}\Rightarrow R_{i,k}, && i\in N,\ k=2,\ldots,m-1, \tag{CSE-2}\\
& X_{i,k}\Rightarrow R_{i,k}, && i\in N,\ k=2,\ldots,m-1, \tag{CSE-3}\\
& X_{i,k}\Rightarrow\neg R_{i,k-1}, && i\in N,\ k=2,\ldots,m-1, \tag{CSE-4}\\
& R_{i,k}\wedge\neg R_{i,k-1}\Rightarrow X_{i,k}, &\qquad& i\in N,\ k=2,\ldots,m-1, \tag{CSE-5}\\
& X_{i,m}\Leftrightarrow\neg R_{i,m-1}, && i\in N. \tag{CSE-5a}
\end{alignat}
For $m=1$, the unit clause $X_{i,1}$ fixes every task at the only station. Station-order precedence is
\begin{equation}
X_{i,k+1}\Rightarrow\neg R_{j,k}
\qquad (i,j)\in E,\ k=1,\ldots,m-1.
\tag{CSE-6}
\end{equation}

If $h_i\geq1$, exactly one start is encoded by
\begin{alignat}{2}
& T_{i,0}\Leftrightarrow S_{i,0}, && i\in N,\ h_i\geq1, \tag{CSE-7}\\
& T_{i,u-1}\Rightarrow T_{i,u}, && i\in N,\ h_i\geq1,\ u=1,\ldots,h_i-1, \tag{CSE-8}\\
& S_{i,u}\Rightarrow T_{i,u}, && i\in N,\ h_i\geq1,\ u=1,\ldots,h_i-1, \tag{CSE-9}\\
& S_{i,u}\Rightarrow\neg T_{i,u-1}, && i\in N,\ h_i\geq1,\ u=1,\ldots,h_i-1, \tag{CSE-10}\\
& T_{i,u}\wedge\neg T_{i,u-1}\Rightarrow S_{i,u}, &\qquad& i\in N,\ h_i\geq1,\ u=1,\ldots,h_i-1, \tag{CSE-11}\\
& S_{i,h_i}\Leftrightarrow\neg T_{i,h_i-1}, && i\in N,\ h_i\geq1. \tag{CSE-11a}
\end{alignat}
If $h_i=0$, the unit clause $S_{i,0}$ fixes task $i$ to start at time zero (CSE-11b).

The following boundary convention avoids separate formulas for values of $u$ outside the allowed start range:
\begin{equation}
\widetilde T_{i,u}=\begin{cases}
\bot,&u<0,\\
T_{i,u},&0\leq u<h_i,\\
\top,&u\geq h_i.
\end{cases}
\qquad i\in N,\ u\in\mathbb Z.
\label{eq:extended-t}
\end{equation}
If a task starts at $s$, it is marked active in every slot that it occupies:
\begin{equation}
S_{i,s}\Rightarrow A_{i,s+\ell}
\quad i\in N,\ s\in H_i,\ \ell=0,\ldots,t_i-1,
\tag{CSE-12}
\end{equation}
and same-station precedence is
\begin{equation}
X_{i,k}\wedge X_{j,k}\wedge S_{i,s}
\Rightarrow\neg\widetilde T_{j,s+t_i-1}
\quad (i,j)\in E,\ k\in K,\ s\in H_i.
\tag{CSE-13}
\end{equation}
If $s+t_i-1\geq h_j$, task $j$ would have to start after its latest allowed start $h_j$. The true boundary value therefore forbids this impossible assignment and start.

The activity-based non-overlap constraint in CSE \citep{VanKieu2025Compact} is
\begin{equation}
X_{i,k}\wedge X_{j,k}\wedge A_{i,\tau}\Rightarrow\neg A_{j,\tau}
\quad i<j,\ i,j\in N,\ k\in K,\ \tau\in H.
\tag{CSE-14}
\end{equation}
At each time slot, the PB power constraint is
\begin{equation}
\sum_{i\in N}w_iA_{i,\tau}\leq\Qmax
\qquad \tau\in H.
\tag{PPC}
\label{eq:sat-power}
\end{equation}
Only this forward activity implication is needed. Setting an extra $A$ variable to true can only make PPC stricter, while a schedule can satisfy the model by setting $A$ true exactly in its occupied slots.

\subsubsection*{Preprocessing and implied precedence arcs}

The CSE forward scan follows precedence order, derives earliest starts, and scans stations from first to last \citep{VanKieu2025Compact}. It rules out earlier starts and stations where the earliest start exceeds $C-t_i$. A backward scan follows reverse precedence order and rules out late starts. Sets $\mathcal F$ and $\mathcal B$ collect these excluded choices:
\begin{align}
&\neg X_{i,k} && (i,k)\in\mathcal F, \tag{CSE-15}\\
&X_{i,k}\Rightarrow\neg S_{i,s} && (i,k,s)\in\mathcal B. \tag{CSE-16}
\end{align}
Both scans use precedence-implied time order and therefore preserve every feasible schedule.

By default, CSE-6 and CSE-13 use $E$. An optional pass processes root tasks (tasks with no predecessors) in input order, skips roots already reached, and adds $(i,k)$ when it finds $i\to j\to k$. The resulting $E^+$ is a sparse subset of implied two-step relations, not the full transitive closure. Because $E$ already implies each addition, feasible schedules are unchanged.

\subsection{Alternative formulations of task non-overlap}

The three formulations differ only in how they prevent overlap. CSE uses CSE-14. The same-station (SS) and same-station disjunction (SS-D) formulations introduce $D_{i,j}$ for $i<j$; this variable is true when tasks $i$ and $j$ share a station. SS compares their activity slots, whereas SS-D requires one task to be before or after the other:
\begin{align}
& X_{i,k}\wedge X_{j,k}\Rightarrow D_{i,j}, && i<j,\ i,j\in N,\ k\in K, \tag{SS-1}\\
& X_{i,k}\wedge X_{j,\ell}\Rightarrow\neg D_{i,j}, && i<j,\ i,j\in N,\ k,\ell\in K,\ k\neq\ell. \tag{SS-2}
\end{align}
The model includes SS-2 for every pair $k\neq\ell$. SS prevents overlap with
\begin{equation}
D_{i,j}\wedge A_{i,\tau}\Rightarrow\neg A_{j,\tau}
\qquad i<j,\ i,j\in N,\ \tau\in H.
\tag{SS-3}
\end{equation}
SS-D replaces SS-3 with an either-before-or-after condition: task $j$ must finish before $i$ starts or start after $i$ finishes.
\begin{equation}
\neg D_{i,j}\vee\neg S_{i,s}\vee
\widetilde T_{j,s-t_j}\vee
\neg\widetilde T_{j,s+t_i-1}
\qquad i<j,\ i,j\in N,\ s\in H_i.
\tag{SS-D}
\end{equation}
Clauses that are automatically true are omitted, and automatically false literals are removed. In particular, $s<t_j$ gives $\widetilde T_{j,s-t_j}=\bot$.

\par\medskip
\subsection{Model correctness and size}

\begin{lemma}
For each task, CSE-1--CSE-5a select exactly one workstation and satisfy $R_{i,k}=1$ if and only if the selected station is at most $k$. Similarly, CSE-7--CSE-11b select exactly one start time and satisfy $T_{i,u}=1$ if and only if the selected start is at most $u$.
\end{lemma}

\begin{proof}
For $m\geq2$, CSE-2 makes $R$ nondecreasing with $k$. An all-true sequence selects station 1, a false-to-true change selects one interior station, and an all-false sequence selects station $m$. CSE-1 and CSE-3--CSE-5a forbid every other $X_{i,k}$. A unit clause covers $m=1$. The same argument applies to $T$ in CSE-7--CSE-11b, including $h_i=0$.
\end{proof}

\begin{lemma}
CSE-15--CSE-16 preserve every feasible schedule.
\end{lemma}

\begin{proof}
Following the tasks in precedence order shows that every value from the forward scan is a valid earliest start. Any earlier start is impossible, and a value above $C-t_i$ makes the station impossible. The same reasoning in reverse precedence order gives the latest-start bounds from the backward scan. Thus, every choice in $\mathcal F$ or $\mathcal B$ is infeasible for cycle time at most $C$.
\end{proof}

\begin{proposition}
For fixed $C$ and $\Qmax$, the formula consisting of CSE-1--CSE-13, one of the three non-overlap blocks, CSE-15--CSE-16, and PPC is satisfiable if and only if a feasible schedule with cycle time at most $C$ exists.
\end{proposition}

\begin{proof}
If the formula is satisfiable, the first lemma gives one station and start time per task. CSE-6 enforces station order, and CSE-13 enforces time order for related tasks at the same station. Each non-overlap formulation prevents simultaneous work at one station. For SS and SS-D, SS-1--SS-2 make $D_{i,j}$ equivalent to sharing a station, and SS-D imposes $\sigma_j+t_j\leq\sigma_i$ or $\sigma_j\geq\sigma_i+t_i$. CSE-12 marks occupied slots, so PPC limits actual power. Preprocessing is safe by the second lemma; hence the schedule is feasible.

Conversely, take a feasible schedule. Set $X$ and $S$ from its assignments and starts, $R$ and $T$ by their definitions, $A$ on occupied slots, and $D_{i,j}$ by station equality. Feasibility and the two lemmas satisfy the formula.
\end{proof}

Table~\ref{tab:complexity} gives the worst-case model sizes, where $\bar t$ is the average task duration. ``New variables'' counts only variables introduced by each block. Because the size depends polynomially on the numerical value of $C$, these bounds are pseudo-polynomial rather than polynomial in the number of bits used to write $C$.

\begin{table}[H]
\centering
\caption{Size of the main SAT encoding blocks.}
\label{tab:complexity}
\footnotesize
\renewcommand{\arraystretch}{1.15}
\setlength{\tabcolsep}{9pt}
\begin{tabular}{@{} l c c @{}}
\toprule
Block & New variables & Clauses \\
\midrule
Assignment prefix & $O(nm)$ & $O(nm)$ \\
Start-time prefix & $O(nC)$ & $O(nC)$ \\
Station precedence & $0$ & $O(|E|m)$ \\
Same-station precedence & $0$ & $O(|E|mC)$ \\
Activity & $O(nC)$ & $O(nC\bar t)$ \\
CSE non-overlap & $0$ & $O(n^2mC)$ \\
SS or SS-D & $O(n^2)$ & $O(n^2m^2+n^2C)$ \\
Station and start pruning & $0$ & $O(nmC)$ \\
\bottomrule
\end{tabular}
\end{table}

Binary Merger converts the $C$ weighted power inequalities into conjunctive normal form (CNF), the SAT solver's clause format. Its data-dependent extra size is not shown. Pairwise exactly-one encodings would require $O(nm^2)$ station and $O(nC^2)$ start clauses; pairwise precedence would require $O(|E|m^2)$ and $O(|E|mC^2)$. Cumulative variables make these parts linear in $m$ or $C$. CSE non-overlap uses $O(n^2mC)$ clauses, while SS-2 uses $O(n^2m^2)$.

\subsection{Search strategies for cycle-time minimization}

The search keeps a feasible upper bound $UB$, a valid lower bound $LB$, and a SAT model for a trial cycle time. The two strategies differ only in how they find their first feasible SAT solution; both then use the same procedure to find shorter cycles.

\subsubsection*{Cycle-time bounds and initial cycle time}
Three simple bounds give the initial cycle-time lower bound
\begin{equation}
\LBzero=\max\left\{
\max_{i\in N} t_i,
\left\lceil\frac{\sum_{i\in N} t_i}{m}\right\rceil,
\left\lceil\frac{\sum_{i\in N} w_it_i}{\Qmax}\right\rceil
\right\}.
\label{eq:cycle-lb}
\end{equation}
The three terms come from the longest task, the average workload per station, and the total task energy allowed by the power limit; summing \eqref{eq:power-cap} gives $\sum_i w_it_i\leq\Qmax C$. If $\Qmax\geq\max_i w_i$, running the tasks one at a time in precedence order at a single station gives the feasible bound
\begin{equation}
C^{\mathrm{seq}}=\sum_{i\in N} t_i.
\label{eq:sequential-ub}
\end{equation}
The search starts from the cycle time
\begin{equation}
\Cinit=\max\left\{\LBzero,
\left\lceil\frac{2\sum_{i\in N} t_i}{m}\right\rceil\right\}.
\label{eq:initial-cycle}
\end{equation}
The factor two raises the initial value above the average workload, while the maximum prevents a trial below $\LBzero$. The guaranteed feasible value is $C^{\mathrm{seq}}$; $\Cinit$ need not be feasible.

\subsubsection*{Initialization strategies}
Direct search makes a full solver call at $\Cinit$ and, if it is infeasible, rebuilds at $C^{\mathrm{seq}}$. Two-phase search first uses calls limited by a fixed number of SAT conflicts while multiplying the tested cycle time by 1.5. Algorithms~\ref{alg:controller} and~\ref{alg:phase-one} show these starts; both then use Algorithm~\ref{alg:incremental-search}.

Let $\phi^*$ denote the best solution found so far, $UB$ its cycle time, $LB$ the current lower bound, $C_{\mathrm{try}}$ the tested value, and $\mathcal S$ the current SAT solver. A satisfiable call returns $\phi$, with $C(\phi)=\max\{s+t_i:S_{i,s}=1\}$.

\begin{algorithm}[H]
\caption{Overall procedure for Direct and Two-phase SAT search}
\label{alg:controller}
\small
\begin{algorithmic}[1]
\Require Instance $(N,E,t,w,m,\Qmax)$; strategy $M\in\{\Status{DIRECT},\Status{TWO\mbox{-}PHASE}\}$; Phase-I conflict limit $B_F$ when $M=\Status{TWO\mbox{-}PHASE}$
\Ensure Best solution $\phi^*$; bounds $[LB,UB]$; proof status
\If{$\Qmax<\max_{i\in N} w_i$} \State \Return $(\varnothing,\LBzero,\infty,\Status{INFEASIBLE})$ \EndIf
\State $LB\gets\LBzero$; $(\phi^*,UB)\gets\textsc{SequentialSchedule}()$; $C_{\mathrm{seq}}\gets UB$
\If{$M=\Status{TWO\mbox{-}PHASE}$}
  \State $(\phi^*,LB,UB,\mathcal S,r)\gets\textsc{PhaseOne}(\phi^*,LB,UB,B_F)$
\Else
  \State $C_{\mathrm{try}}\gets\min\{\Cinit,C_{\mathrm{seq}}\}$
  \State $\mathcal S\gets\textsc{BuildSolver}(\Phi(C_{\mathrm{try}},\Qmax))$; $(r,\phi)\gets\textsc{Solve}(\mathcal S)$
  \If{$r=\Status{UNSAT}$}
    \State $LB\gets\max\{LB,C_{\mathrm{try}}+1\}$
    \State $\mathcal S\gets\textsc{BuildSolver}(\Phi(C_{\mathrm{seq}},\Qmax))$; $(r,\phi)\gets\textsc{Solve}(\mathcal S)$
  \EndIf
  \If{$r=\Status{SAT}$} \State $(\phi^*,UB)\gets(\phi,C(\phi))$ \EndIf
\EndIf
\If{$r=\Status{TIMEOUT}$} \State \Return $(\phi^*,LB,UB,\Status{TIMEOUT})$ \EndIf
\If{$UB=LB$} \State \Return $(\phi^*,UB,UB,\Status{OPTIMAL})$ \EndIf
\State \Return $\textsc{IncrementalOptimize}(\mathcal S,\phi^*,LB,UB)$
\end{algorithmic}
\end{algorithm}

\begin{algorithm}[H]
\caption{Phase I of Two-phase search: feasibility checks with a conflict limit}
\label{alg:phase-one}
\small
\begin{algorithmic}[1]
\Require Best solution $\phi^*$; bounds $[LB,UB]$; conflict limit $B_F$
\Ensure Best solution and bounds; SAT solver $\mathcal S$; status $r$
\State $C_{\mathrm{seq}}\gets UB$; $C_{\mathrm{try}}\gets\min\{\Cinit,C_{\mathrm{seq}}\}$
\While{$C_{\mathrm{try}}<C_{\mathrm{seq}}$}
  \State $\mathcal S\gets\textsc{BuildSolver}(\Phi(C_{\mathrm{try}},\Qmax))$
  \State $(r,\phi)\gets\textsc{SolveLimited}(\mathcal S,B_F)$
  \If{$r=\Status{SAT}$}
    \State \Return $(\phi,LB,C(\phi),\mathcal S,r)$
  \ElsIf{$r=\Status{UNSAT}$}
    \State $LB\gets\max\{LB,C_{\mathrm{try}}+1\}$
  \ElsIf{$r=\Status{TIMEOUT}$}
    \State \Return $(\phi^*,LB,UB,\mathcal S,r)$
  \ElsIf{$r=\Status{UNKNOWN}$} \Comment{keep $LB$ unchanged}
  \EndIf
  \State $C_{\mathrm{try}}\gets\min\{C_{\mathrm{seq}},\max(C_{\mathrm{try}}+1,\lceil1.5C_{\mathrm{try}}\rceil)\}$
\EndWhile
\State $\mathcal S\gets\textsc{BuildSolver}(\Phi(C_{\mathrm{seq}},\Qmax))$
\State $(r,\phi)\gets\textsc{Solve}(\mathcal S)$
\If{$r=\Status{SAT}$} \State \Return $(\phi,LB,C(\phi),\mathcal S,r)$ \EndIf
\State \Return $(\phi^*,LB,UB,\mathcal S,r)$
\end{algorithmic}
\end{algorithm}

\subsubsection*{Incremental improvement and optimality proof}

\begin{samepage}
Any schedule shorter than the current best value $UB$ must satisfy $\sigma_i+t_i\leq UB-1$ for every task. The model adds one Boolean requirement for each task:
\begin{equation}
\Gamma(UB)=\{T_{i,UB-t_i-1}:i\in N\}.
\label{eq:cut}
\end{equation}
\end{samepage}
\begin{algorithm}[H]
\caption{Shared incremental improvement and optimality proof}
\label{alg:incremental-search}
\small
\begin{algorithmic}[1]
\Require SAT solver $\mathcal S$ with a feasible solution; best solution $\phi^*$; bounds $[LB,UB]$
\Ensure Best solution and bounds; proof status
\While{$LB<UB$}
  \State add every literal in $\Gamma(UB)$ to $\mathcal S$
  \State $(r,\phi)\gets\textsc{SolveIncrementally}(\mathcal S)$
  \If{$r=\Status{UNSAT}$} \State \Return $(\phi^*,UB,UB,\Status{OPTIMAL})$ \EndIf
  \If{$r=\Status{TIMEOUT}$} \State \Return $(\phi^*,LB,UB,\Status{TIMEOUT})$ \EndIf
  \State $(\phi^*,UB)\gets(\phi,C(\phi))$
  \If{$UB=LB$} \State \Return $(\phi^*,UB,UB,\Status{OPTIMAL})$ \EndIf
\EndWhile
\State \Return $(\phi^*,UB,UB,\Status{OPTIMAL})$
\end{algorithmic}
\end{algorithm}

A regular solver call returns \Status{SAT} when it finds a feasible schedule, \Status{UNSAT} when it proves that none exists, or \Status{TIMEOUT} when the time limit is reached first. The shorter Phase-I call, \textsc{SolveLimited}, may return \Status{UNKNOWN} after reaching its conflict limit; this result proves neither feasibility nor infeasibility. If the lower and upper bounds become equal, the current solution is optimal.

Because $LB\geq\max_i t_i$ and the improvement loop requires $LB<UB$, each index $u=UB-t_i-1$ in \eqref{eq:cut} is nonnegative. If the model covers times up to $C_h\geq UB$, then $u\leq C_h-t_i-1$, so every required variable exists.

In Phase I, \Status{UNKNOWN} leaves $LB$ unchanged, whereas \Status{UNSAT} proves the tested value infeasible and raises $LB$ to $C_{\mathrm{try}}+1$. Multiplying the tested value by 1.5 reaches $C^{\mathrm{seq}}$ after relatively few model rebuilds. A \Status{SAT} result gives a better solution and a solver that can be reused; otherwise, the one-at-a-time schedule and a solver at $C^{\mathrm{seq}}$ remain available.

The $-1$ in \eqref{eq:cut} is needed because a cycle shorter than $UB$ can end no later than $UB-1$: hence $\sigma_i\leq UB-t_i-1$. Each \Status{SAT} result decreases $UB$. Equality with $LB$, or \Status{UNSAT} after asking for a shorter cycle, proves optimality. A timeout returns the best feasible solution and the valid bounds.

\begin{proposition}
Algorithms~\ref{alg:controller}--\ref{alg:incremental-search} return a feasible solution whenever $\Qmax\geq\max_{i\in N} w_i$. If they return \Status{OPTIMAL}, the reported cycle time is globally optimal.
\end{proposition}

\begin{proof}
The schedule that runs tasks one at a time in precedence order has no overlap and never uses more than the power of one task, which does not exceed $\Qmax$. Thus, the model at $C^{\mathrm{seq}}$ is satisfiable. \Status{UNKNOWN} does not change a bound, and $LB$ increases only after \Status{UNSAT}. Because \eqref{eq:cut} describes every schedule shorter than the current best one, \Status{SAT} decreases $UB$, whereas \Status{UNSAT} excludes all smaller values. Equal bounds, or proof that every smaller value is infeasible, therefore proves global optimality.
\end{proof}

\subsection{Time-indexed MIP and CP formulations}

Each commercial model uses a finite time range, called the horizon and denoted by $\bar C$. If no schedule exists in this range, it is increased by a factor of 1.5; $C^{\mathrm{seq}}=\sum_i t_i$ is a guaranteed feasible final value. Let $\bar H=\{0,\ldots,\bar C-1\}$ and $\bar H_i=\{0,\ldots,\bar C-t_i\}$. Binary variables $x_{i,k}$ and $y_{i,s}$ choose a station and a start time. The activity value is
\begin{equation}
b_{i,\tau}=\sum_{s=\max\{0,\tau-t_i+1\}}^{\min\{\tau,\bar C-t_i\}}y_{i,s}
\qquad i\in N,\ \tau\in\bar H.
\end{equation}
Exactly one station and one start are selected:
\begin{align}
& \sum_{k\in K}x_{i,k}=1, \quad \sum_{s\in\bar H_i}y_{i,s}=1 && i\in N. \label{eq:mip-assign}
\end{align}
For every $(i,j)\in E$ and $k\in K$, station order and same-station temporal order are imposed by
\begin{align}
& x_{j,k}\leq\sum_{h=1}^{k}x_{i,h} && (i,j)\in E,\ k\in K, \label{eq:mip-station-order}\\
& y_{j,s}\leq\sum_{r=0}^{s-t_i}y_{i,r}+2-x_{i,k}-x_{j,k} && (i,j)\in E,\ k\in K,\ s\in\bar H_j, \label{eq:mip-time-order}
\end{align}
where an empty sum is zero. Pairwise station capacity and the peak power constraint are
\begin{align}
& x_{i,k}+x_{j,k}+b_{i,\tau}+b_{j,\tau}\leq 3 && i<j,\ i,j\in N,\ k\in K,\ \tau\in\bar H, \label{eq:mip-station-cap}\\
& \sum_{i\in N}w_i b_{i,\tau}\leq\Qmax && \tau\in\bar H. \label{eq:mip-power-cap}
\end{align}
The integer objective variable $Z$ satisfies
\begin{align}
& Z\geq\sum_{s\in\bar H_i}s y_{i,s}+t_i && i\in N, \label{eq:mip-completion}\\
& Z\geq\sum_{i\in N}t_i x_{i,k} && k\in K. \label{eq:mip-workload}
\end{align}
The model minimizes $Z$, and \eqref{eq:mip-workload} adds a valid lower bound. Every method uses $\LBzero$: SAT initializes $LB$ with it, Gurobi bounds $Z$, and CPLEX adds $Z\geq\LBzero$. The three terms follow from \eqref{eq:mip-completion}, \eqref{eq:mip-workload}, and the sum of \eqref{eq:mip-power-cap}. Gurobi and CPLEX solve the MIP; CPLEX-CP uses the same discrete-time constraints in CP Optimizer.

\section{Experimental design}
\label{sec:experiment}

\subsection{Benchmark set and standard SALBP-2 optima}

The 72 cases combine 13 classical SALBP graphs with the workstation counts in Table~\ref{tab:benchmark-inventory}. Task durations and precedence relations come from the standard files \citep{SchollBenchmark1993}. Following earlier studies \citep{Zheng2024Optimizing,VanKieu2025Compact}, each graph uses one fixed power vector. Each task power is drawn independently and uniformly from the integers 5 to 50. The vectors are provided with the inputs and reused in all tests. The main analysis weights cases equally; a second analysis weights the 13 graphs equally.

\begin{table}[H]
\centering
\caption{Benchmark precedence graphs and exact workstation counts used.}
\label{tab:benchmark-inventory}
\footnotesize
\renewcommand{\arraystretch}{1.15}
\setlength{\tabcolsep}{4pt}
\begin{tabularx}{\textwidth}{@{} l c c Y @{\quad} l c c Y @{}}
\toprule
Graph & $n$ & Cases & Tested $m$ & Graph & $n$ & Cases & Tested $m$ \\
\midrule
BOWMAN & 8 & 1 & 5 & BUXEY & 29 & 6 & 7, 8, 10--13 \\
GUNTHER & 35 & 6 & 7--9, 11, 12, 14 & HESKIA & 28 & 4 & 3--5, 8 \\
JACKSON & 11 & 5 & 3--6, 8 & JAESCHKE & 9 & 4 & 3, 4, 6, 8 \\
LUTZ2 & 89 & 11 & 24--26, 28, 29, 31, 34, 37, 40, 44, 49 & MANSOOR & 11 & 3 & 2--4 \\
MERTENS & 7 & 4 & 2, 3, 5, 6 & MITCHELL & 21 & 3 & 3, 5, 8 \\
ROSZIEG & 25 & 4 & 4, 6, 8, 10 & SAWYER & 30 & 8 & 5, 7, 8, 10--14 \\
WARNECKE & 58 & 13 & 14, 15, 17, 19--25, 27, 29, 31 &  &  &  &  \\
\addlinespace[3pt]
\multicolumn{8}{@{}c}{\textit{Total: 13 precedence graphs and 72 benchmark cases}} \\
\bottomrule
\end{tabularx}
\end{table}

The standard SALBP-2 optimum without a power limit is the baseline. The Scholl SALBP-1 values are feasible upper bounds but do not alone prove SALBP-2 optimality \citep{SchollBenchmark1993}. The SAT run without a power limit is called No-Peak SAT.

\begin{table}[H]
\centering
\caption{Sources of the standard SALBP-2 optima.}
\label{tab:no-peak-reference}
\footnotesize
\renewcommand{\arraystretch}{1.15}
\setlength{\tabcolsep}{6pt}
\begin{tabularx}{\textwidth}{@{} >{\raggedright\arraybackslash}p{5.5cm} c c >{\raggedright\arraybackslash}X @{}}
\toprule
Reference source & Cases & \shortstack{Proved by\\No-Peak SAT} & Basis for optimality \\
\midrule
Scholl SALBP-2 benchmark value & 35 & 31 / 35 & Published exact SALBP-2 optimum \\
Scholl SALBP-1 feasible cycle time + No-Peak SAT & 37 & 37 / 37 & Optimality proved by SAT \\
\midrule
\textit{Total} & 72 & 68 / 72 & All reference values are proven optima \\
\bottomrule
\end{tabularx}
\end{table}

The remaining four cases use their published optima.

\FloatBarrier
\subsection{Peak power limits and SAT configurations}

To set reproducible test limits, order the task powers as $w_1^{\downarrow}\geq\cdots\geq w_n^{\downarrow}$. Because at most one task per workstation can be active, total power lies within these bounds:
\begin{equation}
UB_Q=\sum_{\ell=1}^{m}w_\ell^{\downarrow},
\qquad LB_Q=\max_{i\in N}w_i.
\label{eq:q-bounds}
\end{equation}
Here $LB_Q$ is the largest single-task power and hence a necessary limit. The sum $UB_Q$ is a valid upper bound, although it may not be reached. The test limits are
\begin{align}
Q_{\max}^{\mathrm{top}}&=\left\lfloor\frac{UB_Q+LB_Q}{2}\right\rfloor, \label{eq:q-top}\\
Q_{\max}^{\mathrm{avg}}&=\left\lfloor\frac{m\bar w+LB_Q}{2}\right\rfloor, \label{eq:q-avg}
\end{align}
where $\bar w$ is mean task power. The top-$m$ formula uses the valid bound $UB_Q$; $m\bar w$ is only a reference value and need not bound schedule power. In all 72 cases, the average-based limit is smaller, both limits exceed $LB_Q$, and $Q_{\max}^{\mathrm{avg}}-LB_Q\geq4$.

For comparison across cases, the position of a limit within $[LB_Q,UB_Q]$ is reported as
\begin{equation}
\alpha_Q=\frac{\Qmax-LB_Q}{UB_Q-LB_Q},
\qquad
s_Q=1-\alpha_Q,
\label{eq:limit-tightness}
\end{equation}
Thus, $\alpha_Q$ shows where the limit lies between $LB_Q$ and $UB_Q$, while $s_Q$ is larger for a tighter limit. Neither value is the actual peak power of a schedule.

The experiments test two limits, two searches (Direct and Two-phase), and three non-overlap formulations (CSE, SS, and SS-D). Both searches use $E$; Direct also tests $E^+$. CSE-14 is the earlier formulation \citep{VanKieu2025Compact}. In Table~\ref{tab:sat-configurations}, D and 2P identify the search, and the suffix identifies $E$ or $E^+$.

\begin{table}[H]
\centering
\caption{SAT configurations reported in the main comparisons.}
\label{tab:sat-configurations}
\footnotesize
\renewcommand{\arraystretch}{1.10}
\setlength{\tabcolsep}{14pt}
\begin{tabular}{@{} l c c l @{}}
\toprule
ID & Search & Arcs & Non-overlap block \\
\midrule
D-CSE/$E$ & Direct & $E$ & CSE (CSE-14) \\
D-SS/$E$ & Direct & $E$ & SS (SS-1--SS-3) \\
D-SS-D/$E$ & Direct & $E$ & SS-D (SS-1--SS-2, SS-D) \\
\addlinespace[2pt]
2P-CSE/$E$ & Two-phase & $E$ & CSE (CSE-14) \\
2P-SS/$E$ & Two-phase & $E$ & SS (SS-1--SS-3) \\
2P-SS-D/$E$ & Two-phase & $E$ & SS-D (SS-1--SS-2, SS-D) \\
\addlinespace[2pt]
D-CSE/$E^+$ & Direct & $E^+$ & CSE (CSE-14) \\
D-SS/$E^+$ & Direct & $E^+$ & SS (SS-1--SS-3) \\
D-SS-D/$E^+$ & Direct & $E^+$ & SS-D (SS-1--SS-2, SS-D) \\
\bottomrule
\end{tabular}
\end{table}

Each method ID specifies its search, non-overlap formulation, and precedence set.

The main experiment contains 864 runs: two searches, three formulations on $E$, two limits, and 72 cases. Another 432 Direct runs test $E^+$. Both Direct and Two-phase runs start at $\Cinit$, which is feasible for every case.

\FloatBarrier
\subsection{Computational environment and solver settings}

Experiments used Google Cloud c4-highcpu-8 instances (8 virtual CPUs, 16 GB RAM, Ubuntu 20.04 LTS) and a 3,600-second time limit including model construction and solver calls. CaDiCaL 1.9.5 through PySAT 1.8.dev20 used Binary Merger for PB constraints \citep{Ignatiev2018PySAT,Biere2024CaDiCaL}; commercial versions were Gurobi-MIP 12.0.3, CPLEX-MIP 22.1.1, and CPLEX-CP 22.1.1.

CaDiCaL used one thread; the commercial solvers used default multithreading on eight virtual CPUs. Memory limits were 4 GB for Gurobi, 2,048 MB for the CPLEX-MIP tree, and 16 GB for CPLEX-CP. CPLEX-MIP had to close its optimality gap. Running-time comparisons use these settings.

Two-phase search stops each Phase-I call after $B_F=50{,}000$ SAT conflicts. The same limit is used for every graph, station count, and power limit. Both strategies use the same 3,600-second time limit.

\subsection{Optimality criteria and performance measures}

A cycle time $C$ is proved optimal if $C-1$ is infeasible (\Status{UNSAT}) or $C$ reaches a valid lower bound. Otherwise it is only the best solution found. The best-known cycle time (BKS), $C^{\mathrm{BKS}}$, is the smallest value in the SAT, MIP, and CP results; it is called optimal only when proved. The relative percentage deviation (RPD) of $UB_{\mathrm{sol}}$ from the BKS is
\begin{equation}
\operatorname{RPD}_{\mathrm{BKS}}=
100\frac{UB_{\mathrm{sol}}-C^{\mathrm{BKS}}}{C^{\mathrm{BKS}}}.
\end{equation}
An RPD of zero denotes a BKS match, and a smaller RPD is better.

When the bounds and cycle time become equal, optimality is counted when that solution was found.

Mean RPD is computed only over runs that return a feasible schedule. For PAR-2, a run that proves optimality contributes its running time; every other run contributes 7,200 seconds, twice the time limit. Geometric-mean (GM) time ratios use cases solved to optimality by both methods. Cumulative plots show solved cases over time.

In the tables, Solution found counts feasible solutions, BKS matches means equality with the BKS, and Solved to optimality uses the rules above. Counts are out of 72 unless stated. Boldface marks the best comparable value within each limit, except when every method ties.

All 72 standard SALBP-2 baseline values are proven optima. A reported cycle time difference $\Delta C$ is the increase relative to the corresponding baseline.

\subsection{Research questions}

The experiments address four research questions.

\emph{RQ1.} How do the two peak power limits change the best-known cycle time relative to the standard SALBP-2 optimum, and how does this change vary with $s_Q$?

\emph{RQ2.} How do the selected SAT configurations compare with the MIP and CP solvers in solution quality, cases solved to optimality, and running time?

\emph{RQ3.} How do Direct and Two-phase search compare in solution quality, cases solved to optimality, and running time?

\emph{RQ4.} How do the three ways of preventing overlap affect model size and running time, and how do the methods using $E^+$ compare with those using $E$?

For RQ2, one of the six SAT configurations on $E$ is selected separately under each power limit by applying four rules in order: solve the most cases to optimality, minimize PAR-2, match the BKS in the most cases, and minimize mean RPD. Table~\ref{tab:sat-main-results} reports all six configurations.

Cases are also grouped by task count $n$ and tasks per station $n/m$. The third measure is
\[
\rho_E=\left\lceil\frac{\sum_iw_it_i}{\Qmax}\right\rceil/C^*_0.
\]
The first two are common SALBP difficulty measures \citep{AlvarezMiranda2023Analysis}. The ratio $\rho_E$ compares the power-based lower bound with the standard optimum $C^*_0$. Values above 1 and 1.25 mean that this bound exceeds $C^*_0$ and exceeds it by more than 25\%, respectively; these are not formal difficulty thresholds.

\section{Results and discussion}
\label{sec:results}

The results first quantify the cycle-time effects of the two limits. They then select the SAT configurations for comparison with the commercial solvers, examine the search and non-overlap formulations, and assess performance by benchmark characteristics.

\subsection{RQ1: effect of peak power limits on cycle time}

Table~\ref{tab:cap-impact} compares the BKS under each power limit with the standard SALBP-2 baseline. Under $Q_{\max}^{\mathrm{top}}$, six cases equal the baseline and 66 exceed it; the mean increase is 16.49\% and the median is 14.29\%. Under $Q_{\max}^{\mathrm{avg}}$, all 72 cases exceed the baseline; the mean increase is 61.02\% and the median is 62.20\%. The top-$m$-based limits have $\alpha_Q=0.488$--$0.500$. The average-based limits have $\alpha_Q=0.098$--$0.442$ (median 0.303) and are tighter in every case.

\begin{table}[H]
\centering
\caption{Best-known cycle time increase relative to the standard SALBP-2 optimum.}
\label{tab:cap-impact}
\footnotesize
\renewcommand{\arraystretch}{1.15}
\setlength{\tabcolsep}{3.5pt}
\begin{tabular*}{\textwidth}{@{\extracolsep{\fill}} l c c c c c c c @{}}
\toprule
Peak limit & Cases & \shortstack{BKS proved\\optimal} & \shortstack{Same as\\baseline} & \shortstack{Above\\baseline} & \shortstack{Mean\\$\Delta C$} & \shortstack{Mean $\Delta C$\\(\%)} & \shortstack{Median $\Delta C$\\(\%)} \\
\midrule
$Q^{\mathrm{top}}_{\max}$ & 72 & 58 & 6 & 66 & 7.07 & 16.49 & 14.29 \\
$Q^{\mathrm{avg}}_{\max}$ & 72 & 50 & 0 & 72 & 30.33 & 61.02 & 62.20 \\
\bottomrule
\end{tabular*}
\par\vspace{2pt}\parbox{\textwidth}{\scriptsize\textit{Note:} The relative-tightness ranges are $s_Q=0.500$--$0.512$ for $Q^{\mathrm{top}}_{\max}$ and $s_Q=0.558$--$0.902$ for $Q^{\mathrm{avg}}_{\max}$.}
\end{table}

The average-based BKS is never smaller than the top-$m$-based BKS; Figure~\ref{fig:cap-effect}(a) compares them case by case. Panel~(b) shows that $s_Q$ alone does not explain the increase. The BKS is proved optimal in 58 and 50 cases, respectively; the other values are unproved best solutions. Table~\ref{tab:weighting-sensitivity} compares equal case and graph weighting.

\begin{figure}[H]
\centering
\begin{minipage}{0.49\textwidth}
\centering
\begin{tikzpicture}
\begin{axis}[width=\linewidth,height=57mm,xmin=0,xmax=95,ymin=0,ymax=95,xtick={0,20,40,60,80},ytick={0,20,40,60,80},xlabel={Increase under $Q^{\mathrm{top}}_{\max}$ (\%)},ylabel={Increase under $Q^{\mathrm{avg}}_{\max}$ (\%)},grid=major,grid style={gray!22,densely dotted},tick label style={font=\scriptsize},label style={font=\small}]
\addplot[domain=0:95,samples=2,no marks,gray!60,densely dashed] {x};
\addplot[only marks,mark=*,mark size=1.6pt,mark options={draw=PlotBlue!90!black,fill=PlotBlue,fill opacity=0.6,draw opacity=0.95}] coordinates {(35.294118,47.058824) (5.882353,55.882353) (3.125000,53.125000) (14.285714,60.714286) (11.111111,55.555556) (2.127660,51.063830) (0.000000,53.658537) (6.250000,68.750000) (9.090909,70.454545) (10.000000,62.500000) (4.166667,62.500000) (6.349206,66.666667) (11.111111,77.777778) (1.169591,56.432749) (5.859375,51.562500) (10.243902,58.536585) (23.255814,73.643411) (6.250000,31.250000) (8.333333,41.666667) (10.000000,40.000000) (11.111111,33.333333) (42.857143,42.857143) (7.692308,30.769231) (10.000000,30.000000) (0.000000,25.000000) (33.333333,33.333333) (14.285714,76.190476) (15.000000,80.000000) (15.789474,84.210526) (16.666667,77.777778) (23.529412,82.352941) (25.000000,87.500000) (20.000000,80.000000) (21.428571,78.571429) (23.076923,76.923077) (25.000000,75.000000) (27.272727,72.727273) (0.000000,13.978495) (1.612903,22.580645) (0.000000,18.750000) (0.000000,40.000000) (10.000000,40.000000) (28.571429,28.571429) (50.000000,50.000000) (5.714286,48.571429) (14.285714,61.904762) (14.285714,64.285714) (7.142857,57.142857) (3.125000,43.750000) (4.761905,61.904762) (12.500000,68.750000) (11.764706,58.823529) (16.129032,61.290323) (21.428571,67.857143) (26.923077,69.230769) (24.000000,64.000000) (0.000000,41.538462) (2.127660,53.191489) (7.317073,58.536585) (24.324324,82.882883) (25.000000,82.692308) (27.173913,84.782609) (27.380952,82.142857) (30.379747,84.810127) (30.263158,84.210526) (30.136986,83.561644) (33.333333,85.507246) (34.848485,86.363636) (34.375000,85.937500) (35.000000,83.333333) (37.500000,83.928571) (39.622642,83.018868)};
\end{axis}
\end{tikzpicture}
\par\smallskip\footnotesize (a) Within-case comparison
\end{minipage}
\hfill
\begin{minipage}{0.49\textwidth}
\centering
\begin{tikzpicture}
\begin{axis}[width=\linewidth,height=57mm,xmin=0.48,xmax=0.92,ymin=0,ymax=95,xtick={0.5,0.6,0.7,0.8,0.9},ytick={0,20,40,60,80},xlabel={Relative tightness $s_Q$},ylabel={BKS increase (\%)},yticklabel style={font=\scriptsize},xticklabel style={font=\scriptsize},label style={font=\small},grid=major,grid style={gray!22,densely dotted},legend style={font=\scriptsize,at={(0.02,0.98)},anchor=north west,draw=gray!40,fill=white,fill opacity=0.95}]
\addplot[only marks,mark=*,mark size=1.6pt,mark options={draw=PlotBlue!90!black,fill=PlotBlue,fill opacity=0.7}] coordinates {(0.500000,35.294118) (0.500000,5.882353) (0.500000,3.125000) (0.501241,14.285714) (0.500000,11.111111) (0.500000,2.127660) (0.500000,0.000000) (0.501155,6.250000) (0.500000,9.090909) (0.500914,10.000000) (0.500000,4.166667) (0.501618,6.349206) (0.501425,11.111111) (0.505376,1.169591) (0.500000,5.859375) (0.502732,10.243902) (0.501650,23.255814) (0.507246,6.250000) (0.505051,8.333333) (0.500000,10.000000) (0.500000,11.111111) (0.500000,42.857143) (0.500000,7.692308) (0.500000,10.000000) (0.500000,0.000000) (0.503067,33.333333) (0.500000,14.285714) (0.500000,15.000000) (0.500000,15.789474) (0.500000,16.666667) (0.500000,23.529412) (0.500412,25.000000) (0.500000,20.000000) (0.500000,21.428571) (0.500000,23.076923) (0.500000,25.000000) (0.500000,27.272727) (0.500000,0.000000) (0.500000,1.612903) (0.503546,0.000000) (0.512195,0.000000) (0.500000,10.000000) (0.500000,28.571429) (0.503497,50.000000) (0.505376,5.714286) (0.502857,14.285714) (0.501742,14.285714) (0.500000,7.142857) (0.503650,3.125000) (0.500000,4.761905) (0.500000,12.500000) (0.500000,11.764706) (0.500000,16.129032) (0.500000,21.428571) (0.501529,26.923077) (0.500000,24.000000) (0.500000,0.000000) (0.500000,2.127660) (0.502326,7.317073) (0.500876,24.324324) (0.500000,25.000000) (0.500730,27.173913) (0.500659,27.380952) (0.500000,30.379747) (0.500000,30.263158) (0.500000,30.136986) (0.500554,33.333333) (0.500000,34.848485) (0.500000,34.375000) (0.500000,35.000000) (0.500000,37.500000) (0.500000,39.622642)};
\addlegendentry{$Q^{\mathrm{top}}_{\max}$}
\addplot[only marks,mark=triangle*,mark size=1.9pt,mark options={draw=PlotOrange!90!black,fill=PlotOrange,fill opacity=0.7}] coordinates {(0.614458,47.058824) (0.718391,55.882353) (0.704787,53.125000) (0.694789,60.714286) (0.683721,55.555556) (0.759843,51.063830) (0.743056,53.658537) (0.729792,68.750000) (0.725738,70.454545) (0.714808,62.500000) (0.751880,62.500000) (0.747573,66.666667) (0.740741,77.777778) (0.795699,56.432749) (0.760870,51.562500) (0.737705,58.536585) (0.699670,73.643411) (0.768116,31.250000) (0.717172,41.666667) (0.674603,40.000000) (0.646667,33.333333) (0.593750,42.857143) (0.791667,30.769231) (0.734694,30.000000) (0.654930,25.000000) (0.558282,33.333333) (0.700617,76.190476) (0.698413,80.000000) (0.696360,84.210526) (0.691203,77.777778) (0.689024,82.352941) (0.684774,87.500000) (0.677591,80.000000) (0.670940,78.571429) (0.663539,76.923077) (0.654182,75.000000) (0.641204,72.727273) (0.875000,13.978495) (0.781250,22.580645) (0.737589,18.750000) (0.902439,40.000000) (0.768293,40.000000) (0.619048,28.571429) (0.559441,50.000000) (0.817204,48.571429) (0.748571,61.904762) (0.707317,64.285714) (0.706215,57.142857) (0.810219,43.750000) (0.759259,61.904762) (0.731034,68.750000) (0.702290,58.823529) (0.690141,61.290323) (0.679739,67.857143) (0.669725,69.230769) (0.658046,64.000000) (0.808824,41.538462) (0.752632,53.191489) (0.734884,58.536585) (0.688266,82.882883) (0.685246,82.692308) (0.677372,84.782609) (0.670619,82.142857) (0.667085,84.810127) (0.664663,84.210526) (0.661290,83.561644) (0.658915,85.507246) (0.655650,86.363636) (0.653292,85.937500) (0.647399,83.333333) (0.640255,83.928571) (0.633218,83.018868)};
\addlegendentry{$Q^{\mathrm{avg}}_{\max}$}
\end{axis}
\end{tikzpicture}
\par\smallskip\footnotesize (b) Relative increase versus $s_Q$
\end{minipage}
\caption{Best-known cycle time increases under the two peak power limits. In panel~(a), the average-based value is larger in 68 cases and equal in 4; the dashed line denotes equality. Panel~(b) shows all 144 results: 72 cases under each of the two limits.}
\label{fig:cap-effect}
\end{figure}
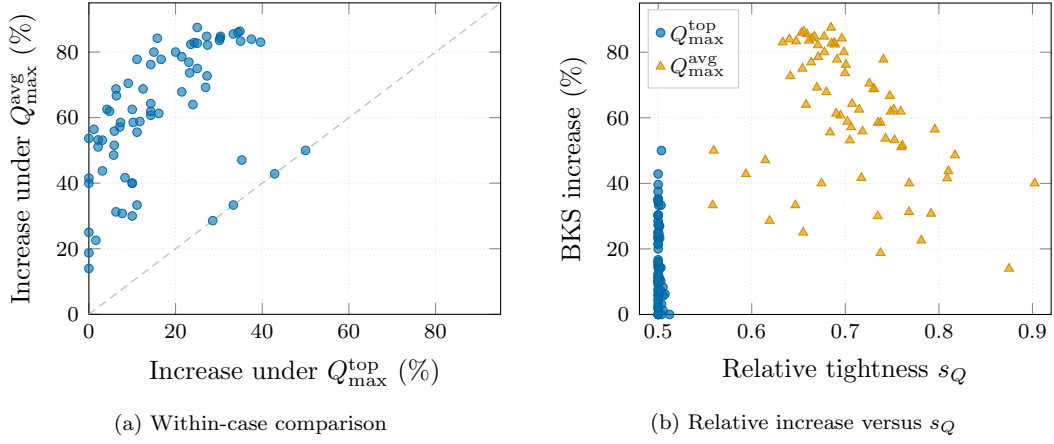

\FloatBarrier
\subsection{SAT configurations used in the solver comparison}

Table~\ref{tab:sat-main-results} reports all six SAT configurations; each finds a feasible solution for all 72 cases. Under $Q_{\max}^{\mathrm{top}}$, the SS and SS-D configurations solve 57 cases to optimality, compared with 56 for CSE. Among the configurations tied at 57, D-SS-D/$E$ has the lowest PAR-2 (1598.6 seconds) and is selected.

\begin{table}[H]
\centering
\caption{Overall performance of all six Direct and Two-phase SAT configurations on $E$.}
\label{tab:sat-main-results}
\footnotesize
\renewcommand{\arraystretch}{1.15}
\setlength{\tabcolsep}{5pt}
\begin{tabular*}{\textwidth}{@{\extracolsep{\fill}} l c c c c @{}}
\toprule
Configuration & \shortstack{BKS\\matches} & \shortstack{Solved to\\optimality} & \shortstack{Mean\\RPD (\%)} & PAR-2 (s) \\
\midrule
\multicolumn{5}{@{}l}{\textit{Top-$m$-based limit ($Q^{\mathrm{top}}_{\max}$)}} \\
\addlinespace[2pt]
D-CSE/$E$ & 69 / 72 & 56 / 72 & 0.09 & 1717.4 \\
2P-CSE/$E$ & 69 / 72 & 56 / 72 & 0.09 & 1767.8 \\
\addlinespace[1pt]
D-SS/$E$ & \textbf{71 / 72} & \textbf{57 / 72} & \textbf{0.01} & 1612.9 \\
2P-SS/$E$ & 70 / 72 & \textbf{57 / 72} & 0.02 & 1624.3 \\
\addlinespace[1pt]
D-SS-D/$E$ & \textbf{71 / 72} & \textbf{57 / 72} & \textbf{0.01} & \textbf{1598.6} \\
2P-SS-D/$E$ & \textbf{71 / 72} & \textbf{57 / 72} & \textbf{0.01} & 1648.1 \\
\addlinespace[4pt]
\multicolumn{5}{@{}l}{\textit{Average-based limit ($Q^{\mathrm{avg}}_{\max}$)}} \\
\addlinespace[2pt]
D-CSE/$E$ & 63 / 72 & 42 / 72 & 0.20 & 3037.5 \\
2P-CSE/$E$ & 62 / 72 & 44 / 72 & 0.18 & 2961.4 \\
\addlinespace[1pt]
D-SS/$E$ & \textbf{65 / 72} & 44 / 72 & 0.19 & 2930.2 \\
2P-SS/$E$ & 63 / 72 & 45 / 72 & 0.12 & 2850.7 \\
\addlinespace[1pt]
D-SS-D/$E$ & 64 / 72 & 45 / 72 & 0.19 & 2851.1 \\
2P-SS-D/$E$ & \textbf{65 / 72} & \textbf{47 / 72} & \textbf{0.09} & \textbf{2766.2} \\
\bottomrule
\end{tabular*}
\end{table}

Under $Q_{\max}^{\mathrm{avg}}$, 2P-SS-D/$E$ solves the most cases to optimality (47). It also has the lowest mean RPD (0.09\%) and PAR-2 (2766.2 seconds), so it is selected.

\FloatBarrier
\subsection{RQ2: comparison with MIP and CP solvers}

The preceding analysis selects D-SS-D/$E$ for $Q_{\max}^{\mathrm{top}}$ and 2P-SS-D/$E$ for $Q_{\max}^{\mathrm{avg}}$. Table~\ref{tab:solver-overall} compares them with the commercial solvers on all 72 cases. Table~\ref{tab:solver-paired} compares running times only when both SAT and the commercial solver prove the same optimum. The next two columns count cases solved to optimality by only one method. Let $t_{\mathrm{solver}}$ and $t_{\mathrm{SAT}}$ denote their running times. A GM ratio $t_{\mathrm{solver}}/t_{\mathrm{SAT}}$ above one means that SAT is faster.

\begin{table}[H]
\centering
\caption{Overall performance of the selected SAT configuration and the commercial MIP and CP solvers.}
\label{tab:solver-overall}
\footnotesize
\renewcommand{\arraystretch}{1.10}
\setlength{\tabcolsep}{4pt}
\begin{tabular*}{\textwidth}{@{\extracolsep{\fill}} l c c c c c @{}}
\toprule
Method & \shortstack{Solution\\found} & \shortstack{BKS\\matches} & \shortstack{Solved to\\optimality} & \shortstack{Mean\\RPD (\%)} & PAR-2 (s) \\
\midrule
\multicolumn{6}{@{}l}{\textit{Top-$m$-based limit ($Q^{\mathrm{top}}_{\max}$)}} \\
\addlinespace[2pt]
Direct SAT (SS-D/$E$) & \textbf{72 / 72} & \textbf{71 / 72} & \textbf{57 / 72} & \textbf{0.01} & \textbf{1598.6} \\
Gurobi-MIP & 44 / 72 & 41 / 72 & 32 / 72 & 0.25 & 4063.4 \\
CPLEX-MIP & 28 / 72 & 25 / 72 & 25 / 72 & 0.76 & 4751.8 \\
CPLEX-CP & 16 / 72 & 16 / 72 & 16 / 72 & -- & 5603.1 \\
\addlinespace[4pt]
\multicolumn{6}{@{}l}{\textit{Average-based limit ($Q^{\mathrm{avg}}_{\max}$)}} \\
\addlinespace[2pt]
Two-phase SAT (SS-D/$E$) & \textbf{72 / 72} & \textbf{65 / 72} & \textbf{47 / 72} & \textbf{0.09} & \textbf{2766.2} \\
Gurobi-MIP & 41 / 72 & 24 / 72 & 22 / 72 & 1.50 & 5049.2 \\
CPLEX-MIP & 33 / 72 & 24 / 72 & 20 / 72 & 3.35 & 5218.6 \\
CPLEX-CP & 14 / 72 & 14 / 72 & 14 / 72 & -- & 5800.1 \\
\bottomrule
\end{tabular*}
\par\vspace{2pt}\parbox{\textwidth}{\scriptsize\textit{Note:} Mean RPD uses the feasible solutions counted in Solution found. It cannot be reported for CPLEX-CP because its nonoptimal solutions are unavailable.}
\end{table}
\begin{table}[H]
\centering
\caption{SAT--solver comparison on cases solved to optimality by both methods; $t_{\mathrm{solver}}/t_{\mathrm{SAT}}>1$ favors SAT.}
\label{tab:solver-paired}
\footnotesize
\renewcommand{\arraystretch}{1.10}
\setlength{\tabcolsep}{4pt}
\begin{tabular*}{\textwidth}{@{\extracolsep{\fill}} l c c c c c @{}}
\toprule
Solver & \shortstack{Both prove\\optimum} & \shortstack{SAT faster\\ / slower} & SAT only & \shortstack{Solver\\only} & \shortstack{GM ratio\\solver / SAT} \\
\midrule
\multicolumn{6}{@{}l}{\textit{Top-$m$-based limit ($Q^{\mathrm{top}}_{\max}$)}} \\
\addlinespace[2pt]
Gurobi-MIP & 31 & 31 / 0 & 26 & 1 & 37.12 \\
CPLEX-MIP & 25 & 25 / 0 & 32 & 0 & 98.34 \\
CPLEX-CP & 16 & 16 / 0 & 41 & 0 & 32.20 \\
\addlinespace[4pt]
\multicolumn{6}{@{}l}{\textit{Average-based limit ($Q^{\mathrm{avg}}_{\max}$)}} \\
\addlinespace[2pt]
Gurobi-MIP & 22 & 21 / 1 & 25 & 0 & 20.52 \\
CPLEX-MIP & 20 & 19 / 1 & 27 & 0 & 35.72 \\
CPLEX-CP & 14 & 14 / 0 & 33 & 0 & 19.04 \\
\bottomrule
\end{tabular*}
\end{table}

\begin{figure}[H]
\centering
\pgfplotslegendfromname{solutionqualitylegend}
\par\smallskip
\begin{minipage}[t]{0.48\textwidth}
\centering
\begin{tikzpicture}
\begin{axis}[xbar stacked,width=57mm,height=38mm,xmin=0,xmax=72,ytick={1,2,3},yticklabels={CPLEX-CP, CPLEX-MIP, Gurobi-MIP},xlabel={Number of cases},xtick={0,18,36,54,72},tick label style={font=\scriptsize},title style={font=\footnotesize\bfseries,yshift=-4pt},title={(a) Top-$m$-based limit ($Q^{\mathrm{top}}_{\max}$)},legend to name=solutionqualitylegend,legend style={font=\scriptsize,legend columns=3,column sep=4pt,draw=gray!30,fill=white,fill opacity=0.95},bar width=10pt,enlarge y limits={abs=0.45},axis line style={gray!50},tick style={gray!50},grid=major,grid style={gray!15,dotted},after end axis/.code={\node[font=\tiny\bfseries,text=black,anchor=center] at (axis cs:8.00,1) {16};\node[font=\tiny\bfseries,text=black,anchor=center] at (axis cs:44.00,1) {56};\node[font=\tiny\bfseries,text=white,anchor=center] at (axis cs:1.50,2) {3};\node[font=\tiny\bfseries,text=black,anchor=center] at (axis cs:15.50,2) {25};\node[font=\tiny\bfseries,text=black,anchor=center] at (axis cs:50.00,2) {44};\node[font=\tiny\bfseries,text=white,anchor=center] at (axis cs:1.50,3) {3};\node[font=\tiny\bfseries,text=black,anchor=center] at (axis cs:23.50,3) {41};\node[font=\tiny\bfseries,text=black,anchor=center] at (axis cs:58.00,3) {28};}]
\addplot+[xbar,fill=PlotBlue!85!white,draw=PlotBlue!90!black,line width=0.4pt] coordinates {(0,1) (3,2) (3,3)};
\addplot+[xbar,fill=gray!25,draw=gray!60,line width=0.4pt] coordinates {(16,1) (25,2) (41,3)};
\addplot+[xbar,fill=PlotOrange!55!white,draw=PlotOrange!90!black,line width=0.4pt] coordinates {(56,1) (44,2) (28,3)};
\legend{SAT better, Same cycle time, Commercial solution unavailable}
\end{axis}
\end{tikzpicture}
\end{minipage}
\hfill
\begin{minipage}[t]{0.48\textwidth}
\centering
\begin{tikzpicture}
\begin{axis}[xbar stacked,width=57mm,height=38mm,xmin=0,xmax=72,ytick={1,2,3},yticklabels={CPLEX-CP, CPLEX-MIP, Gurobi-MIP},xlabel={Number of cases},xtick={0,18,36,54,72},tick label style={font=\scriptsize},title style={font=\footnotesize\bfseries,yshift=-4pt},title={(b) Average-based limit ($Q^{\mathrm{avg}}_{\max}$)},bar width=10pt,enlarge y limits={abs=0.45},axis line style={gray!50},tick style={gray!50},grid=major,grid style={gray!15,dotted},after end axis/.code={\node[font=\tiny\bfseries,text=black,anchor=center] at (axis cs:7.00,1) {14};\node[font=\tiny\bfseries,text=black,anchor=center] at (axis cs:43.00,1) {58};\node[font=\tiny\bfseries,text=white,anchor=center] at (axis cs:4.50,2) {9};\node[font=\tiny\bfseries,text=black,anchor=center] at (axis cs:21.00,2) {24};\node[font=\tiny\bfseries,text=black,anchor=center] at (axis cs:52.50,2) {39};\node[font=\tiny\bfseries,text=white,anchor=center] at (axis cs:8.50,3) {17};\node[font=\tiny\bfseries,text=black,anchor=center] at (axis cs:29.00,3) {24};\node[font=\tiny\bfseries,text=black,anchor=center] at (axis cs:56.50,3) {31};}]
\addplot+[xbar,forget plot,fill=PlotBlue!85!white,draw=PlotBlue!90!black,line width=0.4pt] coordinates {(0,1) (9,2) (17,3)};
\addplot+[xbar,forget plot,fill=gray!25,draw=gray!60,line width=0.4pt] coordinates {(14,1) (24,2) (24,3)};
\addplot+[xbar,forget plot,fill=PlotOrange!55!white,draw=PlotOrange!90!black,line width=0.4pt] coordinates {(58,1) (39,2) (31,3)};
\end{axis}
\end{tikzpicture}
\end{minipage}
\caption{Cycle-time results over all 72 cases. Each row shows cases where SAT finds a lower cycle time, both methods find the same cycle time, or the commercial-solver solution is unavailable in the result files. Labels give case counts, and each row sums to 72. Whenever both solutions are available, SAT never returns a higher cycle time.}
\label{fig:solver-solution-quality}
\end{figure}
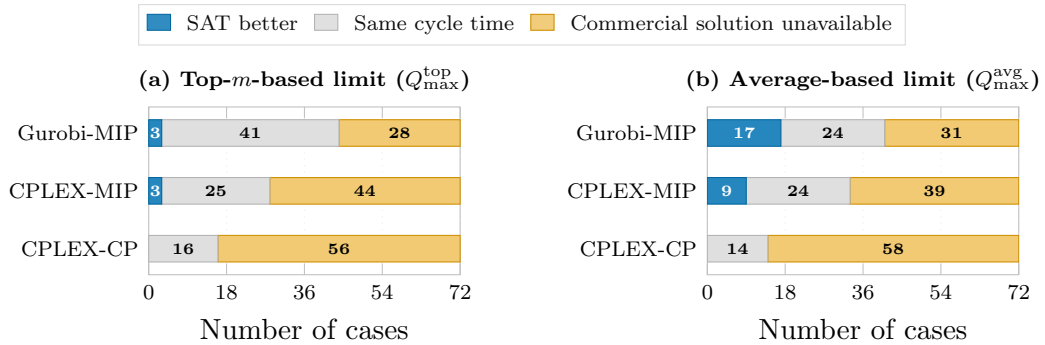

D-SS-D/$E$ matches the BKS in 71 of 72 top-$m$ cases; 2P-SS-D/$E$ does so in 65 of 72 average-based cases. The median RPD is zero under both limits. Figure~\ref{fig:solver-solution-quality} compares cycle times and shows when a commercial-solver solution is unavailable. In particular, CPLEX-CP returns the same cycle time as SAT in all 16 top-$m$ and 14 average-based cases for which its solution is available; its solution is unavailable for the remaining 56 and 58 cases, respectively.

Under $Q_{\max}^{\mathrm{top}}$, SAT is faster on average by a factor of 37.12 against Gurobi-MIP on 31 jointly solved cases, 98.34 against CPLEX-MIP on 25, and 32.20 against CPLEX-CP on 16. Under $Q_{\max}^{\mathrm{avg}}$, the corresponding factors are 20.52 on 22 cases, 35.72 on 20, and 19.04 on 14. SAT is faster in all 72 individual comparisons under $Q_{\max}^{\mathrm{top}}$ and in 54 of 56 under $Q_{\max}^{\mathrm{avg}}$. These counts include only cases solved to optimality by both methods, as listed in Table~\ref{tab:solver-paired}. Because the solvers use different numbers of threads, these are not speedups measured with equal thread counts.

Figure~\ref{fig:cactus} shows how many cases are solved to optimality as the running time increases.

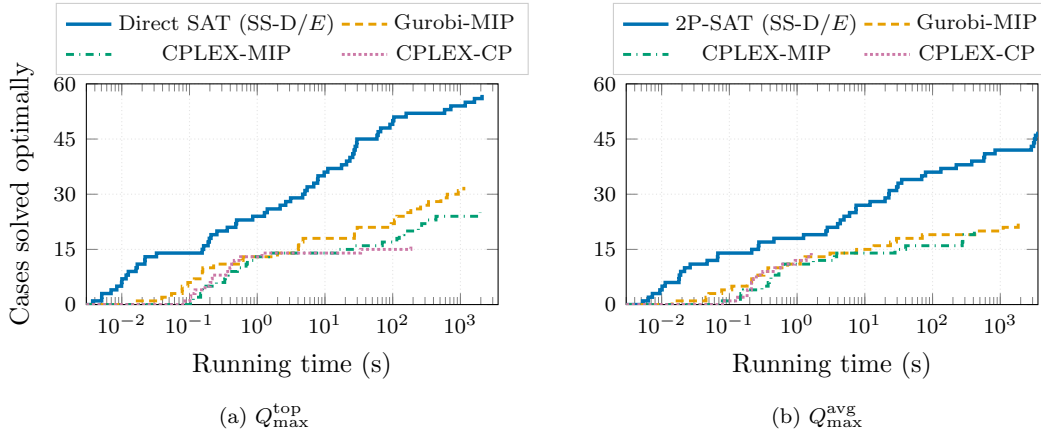
\begin{figure}[H]
\centering
\begin{minipage}{0.49\textwidth}
\centering
\begin{tikzpicture}
\begin{axis}[width=\linewidth,height=45mm,xmode=log,xmin=0.003,xmax=3600,ymin=0,ymax=60,ytick={0,15,30,45,60},xlabel={Running time (s)},ylabel={Cases solved optimally},grid=major,grid style={gray!22,densely dotted},legend style={font=\scriptsize,at={(0.5,1.04)},anchor=south,legend columns=2,draw=gray!40,fill=white,fill opacity=0.95},tick label style={font=\scriptsize},label style={font=\small}]
\addplot[const plot,no marks,color=PlotBlue,solid,line width=1.3pt] coordinates {(0.003,0) (0.003714323,1) (0.005048752,2) (0.005061150,3) (0.006953955,4) (0.008439541,5) (0.009927988,6) (0.010047436,7) (0.011590958,8) (0.012493372,9) (0.016264915,10) (0.016921759,11) (0.021923065,12) (0.022077560,13) (0.032545090,14) (0.153965235,15) (0.161544323,16) (0.188275814,17) (0.196120024,18) (0.207574606,19) (0.253833055,20) (0.361470461,21) (0.482437134,22) (0.495467901,23) (0.857336760,24) (1.278254271,25) (1.406626701,26) (2.181280851,27) (2.631514788,28) (3.110255957,29) (4.682889462,30) (5.061407566,31) (5.454530954,32) (6.827005148,33) (7.788818121,34) (7.885977030,35) (9.745141983,36) (11.012096880,37) (17.578290940,38) (20.827795030,39) (23.344389440,40) (25.687845230,41) (26.903459070,42) (27.935665130,43) (29.387986660,44) (29.612439870,45) (58.555401330,46) (60.958471780,47) (66.386426690,48) (91.038785220,49) (100.011111300,50) (103.876433800,51) (158.042104000,52) (583.363004200,53) (722.255469600,54) (1179.167994000,55) (1611.879743000,56) (2087.566562000,57)};
\addlegendentry{Direct SAT (SS-D/$E$)}
\addplot[const plot,no marks,color=PlotOrange,densely dashed,line width=1.1pt] coordinates {(0.003,0) (0.015736103,1) (0.038957596,2) (0.054211378,3) (0.077969074,4) (0.078094244,5) (0.094443798,6) (0.117602825,7) (0.152279139,8) (0.154502153,9) (0.157175779,10) (0.227361441,11) (0.588612795,12) (0.619363069,13) (1.989364386,14) (4.041596889,15) (4.043785572,16) (4.112289906,17) (4.747301102,18) (27.183960200,19) (27.241860630,20) (29.895427940,21) (94.933160540,22) (105.969540800,23) (115.025219900,24) (153.545307200,25) (185.230606300,26) (258.247517600,27) (338.200370300,28) (514.850428300,29) (637.906011300,30) (930.133920000,31) (1123.260260000,32)};
\addlegendentry{Gurobi-MIP}
\addplot[const plot,no marks,color=PlotGreen,dashdotted,line width=1.1pt] coordinates {(0.003,0) (0.103470325,1) (0.113736629,2) (0.145992041,3) (0.160179615,4) (0.174883127,5) (0.249468327,6) (0.331637383,7) (0.336660147,8) (0.371800423,9) (0.509796143,10) (0.684638023,11) (0.715104818,12) (0.908212662,13) (1.468562365,14) (15.349610570,15) (26.812348370,16) (70.230240350,17) (115.213775600,18) (125.066116800,19) (146.385148800,20) (239.075233700,21) (267.683440700,22) (325.308251400,23) (435.926925400,24) (1953.797598000,25)};
\addlegendentry{CPLEX-MIP}
\addplot[const plot,no marks,color=PlotPurple,densely dotted,line width=1.1pt] coordinates {(0.003,0) (0.085296631,1) (0.101630449,2) (0.104367733,3) (0.120749712,4) (0.169259787,5) (0.202065468,6) (0.206754208,7) (0.216585875,8) (0.344884157,9) (0.394920349,10) (0.418884277,11) (0.458345175,12) (0.567824841,13) (1.174881220,14) (33.635968920,15) (185.564700400,16)};
\addlegendentry{CPLEX-CP}
\end{axis}
\end{tikzpicture}
\par\smallskip\footnotesize (a) $Q^{\mathrm{top}}_{\max}$
\end{minipage}
\hfill
\begin{minipage}{0.49\textwidth}
\centering
\begin{tikzpicture}
\begin{axis}[width=\linewidth,height=45mm,xmode=log,xmin=0.003,xmax=3600,ymin=0,ymax=60,ytick={0,15,30,45,60},xlabel={Running time (s)},grid=major,grid style={gray!22,densely dotted},legend style={font=\scriptsize,at={(0.5,1.04)},anchor=south,legend columns=2,draw=gray!40,fill=white,fill opacity=0.95},tick label style={font=\scriptsize},label style={font=\small}]
\addplot[const plot,no marks,color=PlotBlue,solid,line width=1.3pt] coordinates {(0.003,0) (0.005237341,1) (0.006369829,2) (0.007310152,3) (0.009406328,4) (0.010154724,5) (0.011291504,6) (0.018010616,7) (0.018234015,8) (0.019078970,9) (0.021774769,10) (0.026400328,11) (0.046944141,12) (0.066911459,13) (0.067498446,14) (0.213874578,15) (0.268489361,16) (0.272375822,17) (0.454990625,18) (1.245343924,19) (2.661614418,20) (2.792440176,21) (3.964365244,22) (4.345561504,23) (4.930686712,24) (5.948643684,25) (7.408803463,26) (7.431241274,27) (11.953231810,28) (19.901531220,29) (22.627681020,30) (22.932389020,31) (29.488364700,32) (30.921741490,33) (35.005388740,34) (69.129204990,35) (78.708900450,36) (132.696657700,37) (223.988322300,38) (378.665958900,39) (570.995739000,40) (586.137917000,41) (837.176721100,42) (2874.898671000,43) (3106.631385000,44) (3200.426179000,45) (3308.478336000,46) (3586.001296000,47)};
\addlegendentry{2P-SAT (SS-D/$E$)}
\addplot[const plot,no marks,color=PlotOrange,densely dashed,line width=1.1pt] coordinates {(0.003,0) (0.014233112,1) (0.045152426,2) (0.049777985,3) (0.077031136,4) (0.108666658,5) (0.183076143,6) (0.213433027,7) (0.263217449,8) (0.323802948,9) (0.370296955,10) (0.724528074,11) (1.115903378,12) (1.308657169,13) (3.058547735,14) (7.559943438,15) (13.677929160,16) (23.913820980,17) (29.353069780,18) (70.537279370,19) (443.977200000,20) (1103.504360000,21) (1840.455753000,22)};
\addlegendentry{Gurobi-MIP}
\addplot[const plot,no marks,color=PlotGreen,dashdotted,line width=1.1pt] coordinates {(0.003,0) (0.074887753,1) (0.145617247,2) (0.147617817,3) (0.161184788,4) (0.302835226,5) (0.374069214,6) (0.389798164,7) (0.411795378,8) (0.539725065,9) (0.552820683,10) (0.633949518,11) (1.739472628,12) (3.550043344,13) (3.839728355,14) (27.962858680,15) (40.030714270,16) (274.104711100,17) (280.538741800,18) (290.674826600,19) (414.209457200,20)};
\addlegendentry{CPLEX-MIP}
\addplot[const plot,no marks,color=PlotPurple,densely dotted,line width=1.1pt] coordinates {(0.003,0) (0.084149361,1) (0.132140398,2) (0.182720661,3) (0.184361219,4) (0.211303949,5) (0.212219954,6) (0.225722313,7) (0.250324011,8) (0.311642885,9) (0.468177557,10) (0.575646639,11) (0.937197447,12) (1.363550663,13) (1.621778727,14)};
\addlegendentry{CPLEX-CP}
\end{axis}
\end{tikzpicture}
\par\smallskip\footnotesize (b) $Q^{\mathrm{avg}}_{\max}$
\end{minipage}
\caption{Cases solved to optimality within 3,600 seconds. The top-$m$-based panel uses Direct SAT (SS-D/$E$), and the average-based panel uses 2P-SAT (SS-D/$E$). Across both limits, at least one SAT configuration on $E$ solves 52 case-and-limit pairs that none of the commercial solvers solves within the reported time and memory limits. Both panels use the same axes.}
\label{fig:cactus}
\end{figure}

\FloatBarrier
\subsection{RQ3: Direct versus Two-phase search}

Table~\ref{tab:controller-paired} compares the two search strategies. ``Both prove optimum'' and the faster/slower counts use cases solved by both; Direct only and 2P only use cases solved by one strategy. A GM ratio $t_{\mathrm{D}}/t_{\mathrm{2P}}>1$ favors Two-phase, and the final column compares the cycle times found on all 72 cases.

Both strategies start at $\Cinit$, so the results reflect differences between their search procedures. Under $Q_{\max}^{\mathrm{top}}$, Direct and Two-phase solve the same number of cases to optimality for each encoding. Under $Q_{\max}^{\mathrm{avg}}$, Two-phase solves 44, 45, and 47 cases with CSE, SS, and SS-D, compared with 42, 44, and 45 for Direct. Direct is nevertheless faster on more jointly solved cases for all three encodings under both limits.

\begin{table}[H]
\centering
\caption{Direct--Two-phase comparison; $t_{\mathrm{D}}/t_{\mathrm{2P}}>1$ favors Two-phase.}
\label{tab:controller-paired}
\footnotesize
\renewcommand{\arraystretch}{1.10}
\setlength{\tabcolsep}{3.5pt}
\begin{tabular*}{\textwidth}{@{\extracolsep{\fill}} l c c c c c c @{}}
\toprule
Encoding & \shortstack{Both prove\\optimum} & Direct only & 2P only & \shortstack{2P faster\\ / slower} & \shortstack{GM ratio\\D / 2P} & \shortstack{2P cycle time\\better / same / worse} \\
\midrule
\multicolumn{7}{@{}l}{\textit{Top-$m$-based limit ($Q^{\mathrm{top}}_{\max}$)}} \\
\addlinespace[2pt]
CSE & 56 & 0 & 0 & 21 / 35 & 0.775 & 0 / 72 / 0 \\
SS & 57 & 0 & 0 & 20 / 37 & 0.669 & 0 / 71 / 1 \\
SS-D & 57 & 0 & 0 & 11 / 46 & 0.592 & 0 / 72 / 0 \\
\addlinespace[4pt]
\multicolumn{7}{@{}l}{\textit{Average-based limit ($Q^{\mathrm{avg}}_{\max}$)}} \\
\addlinespace[2pt]
CSE & 41 & 1 & 3 & 18 / 23 & 0.735 & 3 / 65 / 4 \\
SS & 42 & 2 & 3 & 16 / 26 & 0.710 & 2 / 66 / 4 \\
SS-D & 44 & 1 & 3 & 21 / 23 & 0.752 & 4 / 65 / 3 \\
\bottomrule
\end{tabular*}
\end{table}

\FloatBarrier
\subsection{RQ4: effects of the non-overlap formulations}

Table~\ref{tab:direct-structural-configurations} reports all six Direct methods. Under $Q_{\max}^{\mathrm{top}}$, they solve 56 or 57 cases to optimality. Replacing $E$ with $E^+$ solves one more case for CSE and one fewer for SS and SS-D; the BKS matches change from 69, 71, and 71 to 70, 69, and 70, respectively. Under $Q_{\max}^{\mathrm{avg}}$, the numbers solved to optimality change by $+3$, $0$, and $+1$, and the BKS matches change by $+1$, $0$, and $+1$.

\begin{table}[H]
\centering
\caption{Performance of the Direct SAT configurations using $E$ and $E^+$.}
\label{tab:direct-structural-configurations}
\footnotesize
\renewcommand{\arraystretch}{1.10}
\setlength{\tabcolsep}{3pt}
\begin{tabular*}{\textwidth}{@{\extracolsep{\fill}} l c c c c c @{}}
\toprule
Configuration & \shortstack{BKS\\matches} & \shortstack{Solved to\\optimality} & \shortstack{Mean\\RPD (\%)} & PAR-2 (s) & \shortstack{Median\\clauses} \\
\midrule
\multicolumn{6}{@{}l}{\textit{Top-$m$-based limit ($Q^{\mathrm{top}}_{\max}$)}} \\
\addlinespace[2pt]
D-CSE/$E$ & 69 / 72 & 56 / 72 & 0.09 & 1717.4 & 422,124 \\
D-SS/$E$ & \textbf{71 / 72} & \textbf{57 / 72} & \textbf{0.01} & 1612.9 & 241,710 \\
D-SS-D/$E$ & \textbf{71 / 72} & \textbf{57 / 72} & \textbf{0.01} & \textbf{1598.6} & \textbf{237,763} \\
\addlinespace[2pt]
D-CSE/$E^+$ & 70 / 72 & \textbf{57 / 72} & 0.02 & 1609.8 & 425,194 \\
D-SS/$E^+$ & 69 / 72 & 56 / 72 & 0.09 & 1673.3 & 244,906 \\
D-SS-D/$E^+$ & 70 / 72 & 56 / 72 & 0.08 & 1686.9 & 240,959 \\
\addlinespace[4pt]
\multicolumn{6}{@{}l}{\textit{Average-based limit ($Q^{\mathrm{avg}}_{\max}$)}} \\
\addlinespace[2pt]
D-CSE/$E$ & 63 / 72 & 42 / 72 & 0.20 & 3037.5 & 417,904 \\
D-SS/$E$ & \textbf{65 / 72} & 44 / 72 & 0.19 & 2930.2 & 231,377 \\
D-SS-D/$E$ & 64 / 72 & 45 / 72 & 0.19 & 2851.1 & \textbf{227,503} \\
\addlinespace[2pt]
D-CSE/$E^+$ & 64 / 72 & 45 / 72 & \textbf{0.18} & 2859.3 & 420,974 \\
D-SS/$E^+$ & \textbf{65 / 72} & 44 / 72 & 0.20 & 2907.7 & 234,446 \\
D-SS-D/$E^+$ & \textbf{65 / 72} & \textbf{46 / 72} & 0.21 & \textbf{2784.2} & 230,572 \\
\bottomrule
\end{tabular*}
\end{table}

For each benchmark case, let $L_{\mathrm{SS}}$ and $L_{\mathrm{CSE}}$ be the clause counts of the two encodings and let $r_L=100(1-L_{\mathrm{SS}}/L_{\mathrm{CSE}})$ denote SS's relative clause-count reduction. Table~\ref{tab:encoding-paired} reports the changes in clause count and running time.

\begin{table}[H]
\centering
\caption{Clause-count changes and running-time ratios for changes in encoding and precedence arcs.}
\label{tab:encoding-paired}
\footnotesize
\renewcommand{\arraystretch}{1.12}
\setlength{\tabcolsep}{5pt}
\begin{tabular*}{\textwidth}{@{\extracolsep{\fill}} l c c c c @{}}
\toprule
& \multicolumn{2}{c}{$Q^{\mathrm{top}}_{\max}$} & \multicolumn{2}{c}{$Q^{\mathrm{avg}}_{\max}$} \\
\cmidrule(lr){2-3}\cmidrule(l){4-5}
Comparison & \shortstack{Median clause\\change (\%)} & \shortstack{GM time\\ratio} & \shortstack{Median clause\\change (\%)} & \shortstack{GM time\\ratio} \\
\midrule
\multicolumn{5}{@{}l}{\textit{Reduction from changing the non-overlap encoding on $E$}} \\
\addlinespace[2pt]
CSE $\rightarrow$ SS ($r_L$) & 34.53 & 1.32 & 34.78 & 1.09 \\
SS $\rightarrow$ SS-D & 1.90 & 0.98 & 1.98 & 0.98 \\
\addlinespace[4pt]
\multicolumn{5}{@{}l}{\textit{Increase from changing the precedence arcs from $E$ to $E^+$}} \\
\addlinespace[2pt]
CSE & 0.58 & 0.929 & 0.58 & 1.055 \\
SS & 1.08 & 0.937 & 1.12 & 1.004 \\
SS-D & 1.09 & 0.973 & 1.14 & 1.050 \\
\bottomrule
\end{tabular*}
\par\vspace{2pt}\parbox{\textwidth}{\scriptsize\textit{Note:} Clause-count changes are medians over all 72 cases. GM time ratios use cases solved to optimality by both methods. Sample sizes for the top-$m$-based/average-based limits are 56/42 (CSE--SS) and 57/44 (SS--SS-D), 56/42 ($E$--$E^+$ for CSE), 56/43 ($E$--$E^+$ for SS), and 56/44 ($E$--$E^+$ for SS-D). A ratio above one means that the method on the right is faster.}
\end{table}

SS reduces the median clause count by about 35\% relative to CSE, and the GM time ratios show that SS is faster under both limits. SS-D reduces the clause count by about 2\% more, but does not improve running time. The methods using $E^+$ have slightly more clauses than those using $E$; their effect on running time depends on the limit and formulation. The Direct and Two-phase comparison uses $E$, while the separate $E^+$ test uses Direct search. The experiment does not isolate why CSE, SS, and SS-D differ.

\FloatBarrier
\subsection{Performance by benchmark characteristics}

Table~\ref{tab:hardness-strata} groups the 72 cases by benchmark characteristics. The selected methods are D-SS-D/$E$ for $Q^{\mathrm{top}}_{\max}$ and 2P-SS-D/$E$ for $Q^{\mathrm{avg}}_{\max}$. BKS matches measure solution quality, while Solved to optimality counts complete proofs within 3,600 seconds. When $n>50$, the selected SAT method matches 23 of 24 BKS values under $Q^{\mathrm{top}}_{\max}$ and 20 of 24 under $Q^{\mathrm{avg}}_{\max}$, while solving 11 and 10 cases to optimality, respectively.

Across the three tasks-per-station groups, the percentages solved to optimality are 80.0\%, 75.0\%, and 92.9\% under the top-$m$-based limit, and 90.0\%, 66.7\%, and 42.9\% under the average-based limit. Thus, only the average-based percentages fall steadily as tasks per station increase. The percentage also falls as the number of tasks increases under both limits: all 17 cases with $n\leq20$ are solved to optimality, compared with 29 of 31 and 11 of 24 larger cases under the top-$m$-based limit, and 20 of 31 and 10 of 24 under the average-based limit. For $\rho_E>1.25$, 1 of 12 cases is solved to optimality under the top-$m$-based limit and 36 of 61 under the average-based limit; both percentages are lower than for smaller $\rho_E$. Because $\rho_E$ depends on $Q_{\max}$, the two limit settings place different cases in each group. These grouped results describe patterns; they do not show that any one characteristic causes difficulty.

\begin{table}[H]
\centering
\caption{BKS matches and cases solved to optimality by the selected SAT configurations, grouped by benchmark characteristics.}
\label{tab:hardness-strata}
\footnotesize
\renewcommand{\arraystretch}{1.15}
\setlength{\tabcolsep}{4pt}
\begin{tabular*}{\textwidth}{@{\extracolsep{\fill}} l c c c c c c @{}}
\toprule
& \multicolumn{3}{c}{$Q^{\mathrm{top}}_{\max}$} & \multicolumn{3}{c}{$Q^{\mathrm{avg}}_{\max}$} \\
\cmidrule(lr){2-4}\cmidrule(l){5-7}
Case group & Cases & \shortstack{BKS\\matches} & \shortstack{Solved to\\optimality} & Cases & \shortstack{BKS\\matches} & \shortstack{Solved to\\optimality} \\
\midrule
\multicolumn{7}{@{}l}{\textit{Tasks per workstation ($n/m$)}} \\
\addlinespace[2pt]
\quad $n/m\leq2$ & 10 & 10 & 8 & 10 & 10 & 9 \\
\quad $2<n/m\leq4$ & 48 & 47 & 36 & 48 & 45 & 32 \\
\quad $n/m>4$ & 14 & 14 & 13 & 14 & 10 & 6 \\
\addlinespace[4pt]
\multicolumn{7}{@{}l}{\textit{Number of tasks ($n$)}} \\
\addlinespace[2pt]
\quad $n\leq20$ & 17 & 17 & 17 & 17 & 17 & 17 \\
\quad $20<n\leq50$ & 31 & 31 & 29 & 31 & 28 & 20 \\
\quad $n>50$ & 24 & 23 & 11 & 24 & 20 & 10 \\
\addlinespace[4pt]
\multicolumn{7}{@{}l}{\textit{Energy-based lower-bound ratio ($\rho_E$)}} \\
\addlinespace[2pt]
\quad $\rho_E<1$ & 14 & 14 & 14 & 1 & 1 & 1 \\
\quad $1\leq\rho_E\leq1.25$ & 46 & 46 & 42 & 10 & 10 & 10 \\
\quad $\rho_E>1.25$ & 12 & 11 & 1 & 61 & 54 & 36 \\
\addlinespace[2pt]
\midrule
\textit{Total} & 72 & 71 & 57 & 72 & 65 & 47 \\
\bottomrule
\end{tabular*}
\end{table}

\begin{table}[H]
\centering
\caption{Effect of the weighting scheme on the cycle-time and optimality summaries.}
\label{tab:weighting-sensitivity}
\footnotesize
\renewcommand{\arraystretch}{1.12}
\setlength{\tabcolsep}{8pt}
\begin{tabular*}{\textwidth}{@{\extracolsep{\fill}} l c c @{}}
\toprule
Weighting scheme & $Q^{\mathrm{top}}_{\max}$ & $Q^{\mathrm{avg}}_{\max}$ \\
\midrule
\multicolumn{3}{@{}l}{\textit{Mean BKS cycle-time increase (\%)}} \\
\addlinespace[2pt]
\quad Cases weighted equally (baseline) & 16.49 & 61.02 \\
\quad Precedence graphs weighted equally & 14.97 & 53.43 \\
\addlinespace[4pt]
\multicolumn{3}{@{}l}{\textit{Cases solved to optimality by selected SAT (\%)}} \\
\addlinespace[2pt]
\quad Cases weighted equally (baseline) & 79.17 & 65.28 \\
\quad Precedence graphs weighted equally & 90.06 & 76.44 \\
\bottomrule
\end{tabular*}
\par\vspace{2pt}\parbox{\textwidth}{\scriptsize\textit{Note:} Equal graph weighting first averages the case-level values within each of the 13 precedence graphs.}
\end{table}

The average-based limit has the larger mean BKS increase and the smaller percentage solved to optimality under both weighting schemes. The percentages differ because the precedence graphs contribute different numbers of station-count cases.

\FloatBarrier
\subsection{Discussion and limitations}

The average-based limit raises the mean BKS by 61.02\%, compared with 16.49\% under the top-$m$-based limit. Cases with similar values of $s_Q$ can still have different cycle-time increases because workstation assignment, precedence relations, task durations, and power demands act together. The scatter plot cannot separate these effects.

The best-performing SAT configurations in these experiments match more BKS values, solve more cases to optimality, and have smaller mean RPD values than the MIP solvers. CPLEX-CP solutions are available only for cases solved to optimality, so its mean RPD cannot be compared. When both methods prove the same optimum, the GM time ratios favor SAT by 32.20--98.34 under the top-$m$-based limit and 19.04--35.72 under the average-based limit. SAT uses one thread, while the commercial solvers use their default settings for multiple threads.

SS has about one-third fewer median clauses than CSE. SS-D reduces the count by less than 2\% more, but does not consistently reduce running time. Using $E^+$ adds 0.58--1.14\% more median clauses, and its effect on running time varies by limit and formulation. Both strategies start at $\Cinit$. Two-phase solves more cases to optimality under the average-based limit and the same number under the top-$m$-based limit. The experiment does not isolate whether this difference arises from the conflict-limited Phase I or the subsequent incremental search.

Thus, the limit construction affects both cycle time and the difficulty of proving optimality, while clause count alone does not determine running time.

The study uses one fixed power vector per graph, integer task durations, constant power during each task, fixed numbers of workstations, and discrete time. It does not test measured or time-varying power, other power distributions, a variable number of workstations, or continuous time. Moreover, RQ2 uses SAT configurations selected on the same cases used for the solver comparison. The results therefore describe these test cases and do not show how the selected configurations would perform on new cases.

\section{Conclusion}

This study introduces \Problem{}, the first known SALBP formulation to fix both the number of workstations and a line-wide peak power limit while minimizing cycle time. Operationally, it determines the shortest feasible cycle, and hence the highest attainable production rate, without exceeding the prescribed peak power limit. The exact SAT method finds feasible schedules, searches for shorter cycles, and proves optimality when possible.

The experiments show that the tightness of the peak power limit substantially affects both cycle time and the difficulty of proving optimality. The best-performing SAT configurations find a feasible solution for every case, frequently match the best-known cycle time, and, under the reported settings, solve more cases to optimality than each commercial MIP and CP solver. On cases solved to optimality by both SAT and the corresponding commercial solver, SAT is faster in nearly all comparisons. The resulting reproducible benchmark provides a reference for exact optimization in power-aware assembly-line balancing.

The study considers discrete time, fixed task-power vectors, and fixed workstation counts. Peak power constraints also arise in job-shop and flexible job-shop scheduling; adapting the SAT framework to those settings would allow direct comparison with existing scheduling methods. Further extensions within assembly-line balancing may address time-varying power, variable workstation counts, and continuous time.

\section*{Acknowledgements}

The authors thank the maintainers of the Assembly Line Balancing benchmark repository for making the benchmark data publicly available.

\section*{CRediT author contribution statement}

\noindent\emph{Bao Gia Hoang:} Methodology; Software; Validation; Data curation; Visualization; Writing -- original draft.

\noindent\emph{Tuyen Van Kieu:} Conceptualization; Methodology; Formal analysis; Visualization; Writing -- original draft; Writing -- review \& editing.

\noindent\emph{Khanh Van To:} Conceptualization; Investigation; Methodology; Formal analysis; Supervision; Project administration; Writing -- review \& editing.

\section*{Disclosure statement}

The authors report there are no competing interests to declare.

\section*{Declaration of generative AI use}

The authors used OpenAI ChatGPT (GPT-5.6) for language refinement and consistency checks during manuscript preparation. The authors reviewed and verified all outputs and take full responsibility for the accuracy, originality, citations, analyses, and conclusions.

\section*{Funding}

The authors received no specific grant from any funding agency in the public, commercial or not-for-profit sectors for this research.

\bibliographystyle{tfcad}
\bibliography{refs}

\end{document}